\documentclass[aps,pra,onecolumn,notitlepage,superscriptaddress,showpacs,nofootinbib]{revtex4-1}
\usepackage{amsmath}
\usepackage{subfigure}
\usepackage{amssymb}
\usepackage{amsfonts}
\usepackage{enumerate}
\usepackage{mathrsfs}
\usepackage{graphicx}
\usepackage{epstopdf}
\usepackage{enumerate}
\usepackage{changepage}
\usepackage{bm}
\usepackage{multirow}
\usepackage{lipsum}
\usepackage{booktabs}
\usepackage{tikz}
\usepackage{caption}
\usepackage{subcaption}
\usepackage{extarrows}
\usepackage{appendix}
\usepackage[justification=centering]{caption}

\newtheorem{definition}{Definition}
\newtheorem{observation}{Observation}

\newtheorem{proposition}{Proposition}

\def\squareforqed{\hbox{\rlap{$\sqcap$}$\sqcup$}}
\def\qed{\ifmmode\squareforqed\else{\unskip\nobreak\hfil
		\penalty50\hskip1em\null\nobreak\hfil\squareforqed
		\parfillskip=0pt\finalhyphendemerits=0\endgraf}\fi}
\def\endenv{\ifmmode\;\else{\unskip\nobreak\hfil
		\penalty50\hskip1em\null\nobreak\hfil\;
		\parfillskip=0pt\finalhyphendemerits=0\endgraf}\fi}
\newenvironment{proof}{\noindent \textbf{{Proof.~} }}{\qed}
\def\Dbar{\leavevmode\lower.6ex\hbox to 0pt
	{\hskip-.23ex\accent"16\hss}D}
\makeatletter
\def\url@leostyle{%
	\@ifundefined{selectfont}{\def\UrlFont{\sf}}{\def\UrlFont{\small\ttfamily}}}
\makeatother
	\newcommand{\nc}{\newcommand}
	
	\def\i{\iota}
	
	\def\l{\lambda}

	\def\o{\omega}

	\nc{\bbA}{\mathbb{A}} \nc{\bbB}{\mathbb{B}} \nc{\bbC}{\mathbb{C}}
	\nc{\bbD}{\mathbb{D}} \nc{\bbE}{\mathbb{E}} \nc{\bbF}{\mathbb{F}}
	\nc{\bbG}{\mathbb{G}} \nc{\bbH}{\mathbb{H}} \nc{\bbI}{\mathbb{I}}
	\nc{\bbJ}{\mathbb{J}} \nc{\bbK}{\mathbb{K}} \nc{\bbL}{\mathbb{L}}
	\nc{\bbM}{\mathbb{M}} \nc{\bbN}{\mathbb{N}} \nc{\bbO}{\mathbb{O}}
	\nc{\bbP}{\mathbb{P}} \nc{\bbQ}{\mathbb{Q}} \nc{\bbR}{\mathbb{R}}
	\nc{\bbS}{\mathbb{S}} \nc{\bbT}{\mathbb{T}} \nc{\bbU}{\mathbb{U}}
	\nc{\bbV}{\mathbb{V}} \nc{\bbW}{\mathbb{W}} \nc{\bbX}{\mathbb{X}}
	\nc{\bbZ}{\mathbb{Z}}
	
	\nc{\bA}{{\bf A}} \nc{\bB}{{\bf B}} \nc{\bC}{{\bf C}}
	\nc{\bD}{{\bf D}} \nc{\bE}{{\bf E}} \nc{\bF}{{\bf F}}
	\nc{\bG}{{\bf G}} \nc{\bH}{{\bf H}} \nc{\bI}{{\bf I}}
	\nc{\bJ}{{\bf J}} \nc{\bK}{{\bf K}} \nc{\bL}{{\bf L}}
	\nc{\bM}{{\bf M}} \nc{\bN}{{\bf N}} \nc{\bO}{{\bf O}}
	\nc{\bP}{{\bf P}} \nc{\bQ}{{\bf Q}} \nc{\bR}{{\bf R}}
	\nc{\bS}{{\bf S}} \nc{\bT}{{\bf T}} \nc{\bU}{{\bf U}}
	\nc{\bV}{{\bf V}} \nc{\bW}{{\bf W}} \nc{\bX}{{\bf X}}
	\nc{\ba}{{\bf a}} \nc{\be}{{\bf e}} \nc{\bu}{{\bf u}}
	\nc{\brr}{{\bf r}} \nc{\bx}{{\bf x}} \nc{\bz}{{\bf z}} \nc{\bv}{{\bf v}} \nc{\bt}{{\bf t}} \nc{\bs}{{\bf s}}

	\nc{\cA}{{\cal A}} \nc{\cB}{{\cal B}} \nc{\cC}{{\cal C}}
	\nc{\cD}{{\cal D}} \nc{\cE}{{\cal E}} \nc{\cF}{{\cal F}}
	\nc{\cG}{{\cal G}} \nc{\cH}{{\cal H}} \nc{\cI}{{\cal I}}
	\nc{\cJ}{{\cal J}} \nc{\cK}{{\cal K}} \nc{\cL}{{\cal L}}
	\nc{\cM}{{\cal M}} \nc{\cN}{{\cal N}} \nc{\cO}{{\cal O}}
	\nc{\cP}{{\cal P}} \nc{\cQ}{{\cal Q}} \nc{\cR}{{\cal R}}
	\nc{\cS}{{\cal S}} \nc{\cT}{{\cal T}} \nc{\cU}{{\cal U}}
	\nc{\cV}{{\cal V}} \nc{\cW}{{\cal W}} \nc{\cX}{{\cal X}}
	\nc{\cZ}{{\cal Z}}
	
	\nc{\hA}{{\hat{A}}} \nc{\hB}{{\hat{B}}} \nc{\hC}{{\hat{C}}}
	\nc{\hD}{{\hat{D}}} \nc{\hE}{{\hat{E}}} \nc{\hF}{{\hat{F}}}
	\nc{\hG}{{\hat{G}}} \nc{\hH}{{\hat{H}}} \nc{\hI}{{\hat{I}}}
	\nc{\hJ}{{\hat{J}}} \nc{\hK}{{\hat{K}}} \nc{\hL}{{\hat{L}}}
	\nc{\hM}{{\hat{M}}} \nc{\hN}{{\hat{N}}} \nc{\hO}{{\hat{O}}}
	\nc{\hP}{{\hat{P}}} \nc{\hR}{{\hat{R}}} \nc{\hS}{{\hat{S}}}
	\nc{\hT}{{\hat{T}}} \nc{\hU}{{\hat{U}}} \nc{\hV}{{\hat{V}}}
	\nc{\hW}{{\hat{W}}} \nc{\hX}{{\hat{X}}} \nc{\hZ}{{\hat{Z}}}
	
	\nc{\hn}{{\hat{n}}}
	
	\def\dim{\mathop{\rm Dim}}

	\def\ghz{\mathop{\rm GHZ}}

	\def\max{\mathop{\rm max}}
	\def\min{\mathop{\rm min}}

	\def\tr{\mathop{\rm Tr}}

	\newcommand{\bra}[1]{\langle#1|}
	\newcommand{\ket}[1]{|#1\rangle}
	
	\newcommand{\ketbra}[2]{|#1\rangle\!\langle#2|}

	\newcommand{\fc}[2]{\left\lceil\frac{#1}{#2}\right\rceil}

	\usepackage[
	colorlinks,
	linkcolor = blue,
	citecolor = blue,
	urlcolor = blue]{hyperref}
	\def \qed {\hfill \vrule height7pt width 7pt depth 0pt}
	
	\newcounter{lastnote}

\begin{document}
\title{Detecting entanglement and nonlocality with minimum observable length}

\author{Zhuo Chen}
\affiliation{Institute for Interdisciplinary Information Sciences, Tsinghua University, Beijing, 100084 China}

\author{Fei Shi}
\affiliation{QICI Quantum Information and Computation Initiative, School of Computing and Data Science,
The University of Hong Kong, Pokfulam Road, Hong Kong}
\author{Qi Zhao}
\email[]{zhaoqi@cs.hku.hk}
\affiliation{QICI Quantum Information and Computation Initiative, School of Computing and Data Science,
The University of Hong Kong, Pokfulam Road, Hong Kong}	

\begin{abstract}
Quantum entanglement and nonlocality are foundational to quantum technologies, driving quantum computation, communication, and cryptography innovations. To benchmark the capabilities of these quantum techniques, efficient detection and accurate quantification methods are indispensable. This paper 
focuses on the concept of ``detection length"--—a metric that quantifies the extent of measurement globality required to verify entanglement or nonlocality. We extend the detection length framework to encompass various entanglement categories and nonlocality phenomena, providing a comprehensive analytical model to determine detection lengths for specified forms of entanglement. Furthermore, we exploit semidefinite programming techniques to construct entanglement witnesses and Bell's inequalities tailored to specific minimal detection lengths, offering an upper bound for detection lengths in given states. By assessing the noise robustness of these witnesses,
we demonstrate that witnesses with shorter detection lengths can exhibit superior performance under certain conditions.
\end{abstract}

\maketitle

\section{Introduction}
Quantum entanglement and nonlocality are recognized as essential resources in various quantum technologies, including but not limited to the quantum computational speed-up \cite{jozsa2003role, Troels2014speedup}, quantum communication \cite{Ursin2007commu, Buhrman2010commu}, and quantum cryptography \cite{Ekert1991crypto, Scarani2009qkd, BENNETT2014qkd}. These phenomena are significant for the realization of a large-scale, functional quantum computer \cite{NISQ}. Various classes of entanglement together with nonlocality have attracted significant interest in recent research, particularly for multipartite systems \cite{Hirsch2016bell, Jung2011Taming,Chitambar2021non,Zhu2024non}.

For the practical application of quantum entanglement and nonlocality, it is crucial to efficiently detect and accurately quantify these quantum resources \cite{TERHAL2002detect,Tura2014non,Miklin2016multi,Chiara2018ent,Zou2021ent}.
Consider a cryptography task involving a multi-party quantum system as an example model, where each party possesses a component of the system and does not trust the others. Assume that tasks, including decoding information from a quantum system \cite{Zanardi1997code,duclos2010fast}, sharing secret quantum information \cite{hillery1999quantum,Gottesman2000Qss}, or verifying the structure of a shared quantum network \cite{Gottesman2001sign, Kliesch2021cert}, requires detecting entanglement properties within the system.
Then a natural question emerges: What is the minimum number of parties required to simultaneously participate in each measurement to detect entanglement of the whole system and detailed structure? The answer to this question depends on the characteristics of the entangled quantum system:
some types of entanglement may need global measurements that engage the entire system simultaneously. In contrast, others may only require measurements that involve a small number of parties for detection \cite{Gittsovich2010cor,Sawicki2012det,Sperling2013wit,Paraschiv2018marg,Tabia2022ent,shi2024entanglement}. There is substantial interest in characterizing entanglement or nonlocality with the minimal required parties. Previous literature has presented several case studies \cite{Walter2013marginal,tura2014localW, Lu2018mul,Gerke2018num,Navascues2021entanglement, Tabia2022transit}. However, a general analytical framework remains to be proposed.

In practice, during the experimental implementation of entanglement detection, the impact of environmental noise on the detection outcomes is ubiquitous, particularly within multipartite systems.
In global measurements involving all the parties, an error in even a single party can result in a wrong outcome. 
On the contrary, we demonstrate that local measurements are generally more resistant to noise compared to the all-encompassing global measurements. While local measurements offer a more reliable method for assessing entanglement and nonlocality, local few-body measurements may fail to capture the global information of the underlying states and thus fail to detect them. A profound comprehension of the trade-off between detection capability and noise robustness is crucial for enhancing the efficiency of experimental protocols that involve entanglement and nonlocality.

In this context, the concept of ``detection length" emerges as a critical metric \cite{shi2024entanglement}, which represents the degree of ``globality'' of the measurements required to ascertain the presence of entanglement or nonlocality.
Our paper extends the detection length framework to encompass diverse entanglement categories and nonlocality phenomena, clarifying the relationships among different detection lengths. Utilizing our analytical model, the detection length for any given form of entanglement can be determined, providing a straightforward metric for entanglement.
Furthermore, we propose novel semidefinite programming (SDP) techniques to construct entanglement witnesses and Bell's inequalities with specific minimal detection lengths, which is also capable of providing an upper bound of detection length for given states. Incorporating global depolarizing noise and bit-flip error models, we assess the noise robustness across various witnesses, finding that those with shorter detection lengths exhibit superior performance under certain conditions.
Collectively, these contributions enhance our understanding of quantum entanglement and nonlocality by offering refined quantitative measures for assessing the presence and magnitude of these quantum properties.

\section{Theoretical formulations}

We denote $\cH:=\otimes_{i=1}^n \cH_i$  as the $n$-partite Hilbert space, and $\cD$ as the set of all density matrices on  $\cH$. Let $S$ be a proper subset of $[n]:=\{1,2,\ldots,n\}$, then we denote $\cH_{S}:=\otimes_{j\in S} \cH_{j}$, and $\rho_S$ as the marginal of $\rho\in \cD$ on $S$, i.e., $\rho_S:=\tr_{\overline{S}}\rho$, where $\overline{S}=[n]\setminus S$.
A pure state $\ket{\psi}\in \cH$ is \emph{product state} if it can be written as $\ket{\psi}=\ket{\alpha}_{1}\otimes \ket{\beta}_{2}\otimes\cdots\otimes \ket{\gamma}_{n}$, where $\ket{*}_i\in \cH_i$ for $1\leq i\leq n$. 
A pure state $\ket{\psi}\in \cH$ is \emph{biproduct state with respect to the bipartition $S|\overline{S}$, or biproduct state on $S|\overline{S}$} for short, if it can be written as $\ket{\psi}=\ket{\phi}_S\otimes \ket{\varphi}_{\overline{S}}$, where 
$\ket{\phi}_S\in \cH_{S}$, and $\ket{\varphi}_{\overline{S}}\in \cH_{\overline{S}}$.
A mixed state $\rho\in \cD$ is called  \emph{fully separable} (\emph{biseparable on $S|\overline{S}$}) if it can be written as a mixture of product states (biproduct states on $S|\overline{S}$), i.e., $\rho=\sum_{i}p_i\ketbra{\psi_i}{\psi_i}$, where each $\ket{\psi_i}$ is a product state (biproduct state on $S| \overline{S}$). A mixed state $\rho\in \cD$ is called \emph{biseparable}
if it can be written as a mixture of biproduct states, where each constituent may be biproduct with respect to different bipartitions. A state $\rho\in \cD$ is called  \emph{entangled (entangled on $S|\overline{S}$)} if it is not fully separable (biseparable on $S|\overline{S}$), and  $\rho$ is called \emph{genuinely entangled} if it is not biseparable. We denote $\cS_k$ as the set of all $k$-subsets of $[n]$, i.e., $|\cS_k|=\binom{n}{k}$, and $|S|=k$ for every $S\in \cS_k$.

Given a state $\rho\in \cD$ and a set of subsets $\cS=\{S_1,S_2,\cdots, S_k\}$ with $S_i\subset [n]$, the \emph{compatibility set} is defined as the collection of all density matrices $\sigma\in\mathcal{D}$ that share the same marginals with $\rho$ on every subset $S_i\in\cS$:
\begin{equation}
  \cC(\rho,\cS)=\{\sigma\in \cD\mid \sigma_{S_i}=\rho_{S_i}, \ \forall \ 1\leq i\leq k\}.
\end{equation}
If certain constraints are imposed on the compatibility set, then the properties of 
$\rho$ can be inferred from its marginals. For example, if $\cC(\rho, \mathcal{S})$ contains only genuinely entangled states, then we say $\cS$ can detect $\rho$'s GME, indicating that measurements on the subsystems in $\cS$ are sufficient to identify $\rho$'s GME \cite{shi2024entanglement}. The GME detection length is defined as \cite{shi2024entanglement}:
\begin{equation}   l_{GME}(\rho):=\min_{\cS} \big\{\max(\cS) ~\big|~ \text{$\cS$ detects $\rho$'s GME} \big\} \, ,
\end{equation}
where $\max{(\cS)}:=\max_{S\in \cS} |S|$ is the maximum size of a
subset in $\cS$. In other words, the GME detection length of a genuinely entangled state quantifies the minimum observable length necessary to detect GME from this state. Similarly, we can define detection lengths for other classes of entanglement.

\begin{definition}\label{def:edl}
    We say that \emph{$\cS$ detect $\rho$'s entanglement (entanglement  on $G|\overline{G}$)} if the compatibility set $\cC(\rho, \mathcal{S})$ contains only 
    entangled states (entangled states on $G|\overline{G}$).  The \emph{entanglement detection length}  is defined as 
\begin{equation}   
l_{Ent}(\rho):=\min_{\cS} \big\{\max(\cS) ~\big|~ \text{$\cS$ detects $\rho$'s entanglement} \big\} \, ,
\end{equation}
and the \emph{entanglement detection length on $G|\overline{G}$} is defined as 
\begin{equation}
  l_{Bipa}(\rho,G):=\min_{\cS}\big\{\max(\cS)~\big|~\text{$\cS$ detects $\rho$'s entanglement on $G|\overline{G}$} \big\} \, .
\end{equation}  
\end{definition}

In particular, for a fully separable state (biseparable state on $G|\overline{G}$, or biseparable state), we denote $l_{Ent}(\rho):=+\infty$ ($l_{Bipa}(\rho,G):=+\infty$, or $l_{GME}(\rho):=+\infty$).
For an entangled state $\rho$, it is impossible to detect its entanglement by using only $k$-body measurement with $k<l_{Ent}(\rho)$, and every experimental scheme that detects $\rho$'s entanglement must include at least one $k$-body measurement with $k\geq l_{Ent}(\rho)$. This principle extends analogously to bipartition entanglement.

The concept of detection length can be directly extended to the nonlocality detection length.  We consider an $n$-partite Bell's scenario, where $n$ distant observers collectively share an $n$-partite state $\rho$. Assume each observer $j$ measures one of two dichotomic observables $A_{j}^{(x_j)}$ with $x_j=0,1$ for $1\leq j\leq n$. Given a subset $S\subset [n]$, the correlations associated with $S$ are described as the expectation values:
\begin{equation}
\tr\left(A^{(x_{j_1})}_{j_1}\otimes A^{(x_{j_2})}_{j_2}\otimes\cdots \otimes A^{(x_{j_m})}_{j_m} \rho\right), \  \text{where} \  S=\{j_1,j_2,\cdots,j_m\} \,. 
\end{equation}
Then Bell's inequalities associated with a set of subsets $\cS$ can be formulated as
\begin{equation}
    \cB(\cS):=\sum_{S\in\cS} c_{j_1,j_2,\cdots, j_m}^{x_{j_1},x_{j_2},\cdots,x_{j_m}} \tr\left(A^{(x_{j_1})}_{j_1}\otimes A^{(x_{j_2})}_{j_2}\otimes\cdots \otimes A^{(x_{j_m})}_{j_m} \rho\right) \leq \beta_C,
\end{equation}
where $\beta_C$ is the classical bound.

\begin{definition}
 We say that \emph{$\cS$ detects $\rho$'s nonlocality} if $\rho$ violates at least one Bell's inequality in the form of $\cB(\cS)$. The \emph{nonlocality detection length} is defined as
\begin{equation}
  l_{Nol}(\rho)=\min_{\mathcal{S}}\{\max{(\cS)}\mid \text{$\cS$ detects $\rho$'s nonlocality}\}\, .
\end{equation}
\end{definition}

The nonlocality detection length describes the minimum length of the Bell's operators needed to detect Bell's nonlocality in a given state. In particular, for a local state $\rho$ (which does not violate any Bell's equalities), we denote  $l_{Nol}(\rho):=+\infty$.
In the next section, we will investigate the relations between various detection lengths.

\section{The relationship between various detection lengths}

The relationship between these categories of entanglement and nonlocality is depicted in Fig.~\ref{fig:ent_structure}.
A genuinely entangled state must be an entangled state on $G|\overline{G}$ for every $G\subset [n]$, and an entangled state on $G|\overline{G}$ must be an entangled state.
According to these interactions, we could obtain the relationship between various entanglement detection lengths:
\begin{equation}\label{eq:entanglement}
 2\leq l_{Ent}(\rho)\leq l_{Bipa}(\rho,G)\leq l_{GME}(\rho), \ \forall \, G\subset [n],    
\end{equation}
where the lower bound is obtained from the fact that $\cC(\rho, \cS_1)$ contains a fully separable state $\rho_{\{1\}}\otimes \rho_{\{2\}}\otimes\cdots\otimes\rho_{\{n\}}$. The entanglement detection length and nonlocality detection length must satisfy: 
\begin{equation}\label{eq:nonlocality}
    2\leq l_{Ent}(\rho)\leq l_{Nol}(\rho). 
\end{equation}
This is because if $\cS$ detects $\rho$'s nonlocality, then $\cS$ can also detect $\rho$'s entanglement. 


\begin{figure}[htbp]
  \centering
  \includegraphics[width=8.6cm]{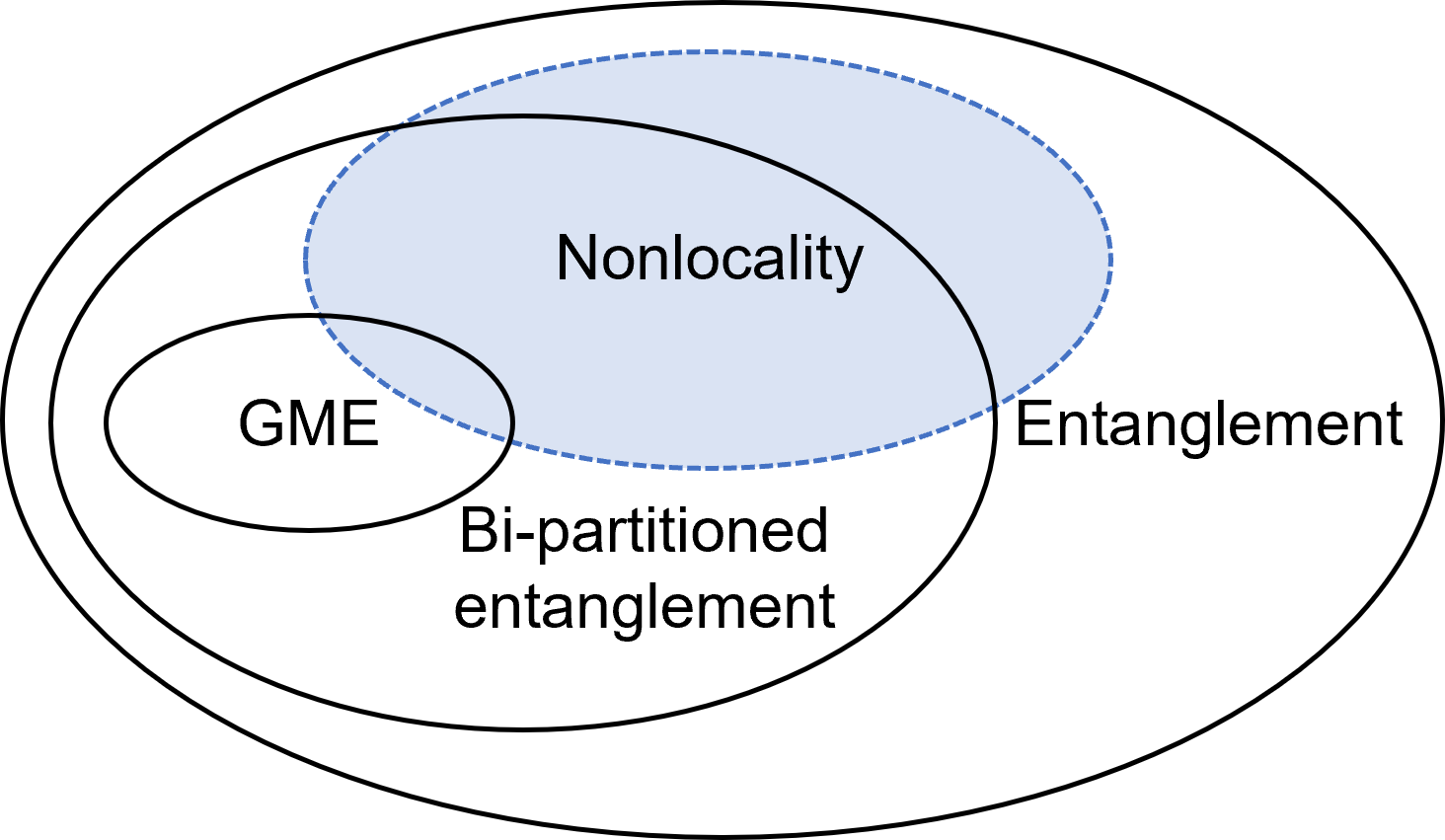}
  \captionsetup{justification=raggedright}
  \caption{The relationship between various forms of entanglement and nonlocality. It straightforwardly implies the relationship among these corresponding detection lengths.}
  \label{fig:ent_structure}
\end{figure}






Prior investigations have established that for certain classes of quantum states, including GHZ states, symmetric states, cluster states, and ring states, there is no gap between the entanglement detection length and the GME detection length 
 \cite{shi2024entanglement}.
Consider the class of graph state as an example. For a graph $G$ with $n$-vertices ($n\geq3$), the associated graph state $\ket{G}$ is defined as the simultaneous $+1$-eigenstate of $n$ operators $\{g_i=X_i\otimes \bigotimes_{j\in N_G(i)}Z_j\}^{n}_{i=1}$, where $N_G(i)$ is the set of all vertices adjacent to $i$. The graph state $\ket{G}$ is genuinely entangled for connected graph $G$ \cite{Hein2004graph}, and the compatibility set $\cC(\ket{G},\cS_2)$ contains a fully separable state \cite{Gittsovich2010multi}. Consequently, for graph states, both the entanglement and GME detection length are at least $3$. Specifically, for an $n$-vertices ring graph,  we can construct a $3$-length witnesses to detect both properties in the ring state $\ket{Ring_n}$, as detailed in Section \ref{sec:SDP}, thus $l_{Ent}(\ket{Ring_n})=l_{GME}(\ket{Ring_n})=3$. 
This observation raises the question: Is this equivalence of detection lengths a universal property for all quantum states? Our finding provides a negative answer, demonstrating the existence of a genuinely entangled state that exhibits the maximal gap between the detection lengths of entanglement and GME.

\begin{proposition}\label{prop:ent_GME_gap}
Let $\ket{\psi_n}=\frac{1}{\sqrt{2}}(\ket{100\cdots 0}+\ket{010\cdots 0})$, $\ket{\ghz_n}=\frac{1}{\sqrt{2}}(\ket{0\cdots 00}+\ket{1\cdots 11})$ and 
\begin{equation*}
\rho=p\ketbra{\psi_n}{\psi_n}+(1-p)\ketbra{{\ghz}_n}{{\ghz}_n},  
\end{equation*}
where $\frac{1}{2}<p<1$, we have $l_{GME}(\rho)=n>l_{Ent}(\rho)=2$. And for any bipartition $G\subset [n]$, if $G$ contains only one of $1$ and $2$, $l_{Bipa}(\rho,G)=2$; otherwise $l_{Bipa}(\rho,G)=n$.
\end{proposition}

The technical proof of this proposition, also of others in the subsequent text, can be found in the Appendix \ref{appdix:proof}.
For the detection lengths of entanglement and nonlocality, there are also gaps. For example, 
some Werner states are entangled but violate no Bell's inequality \cite{Huber2022werner,Eggeling2000Werner,Bowles2016GME_LHV}. 
Therefore, these $n$-qubit Werner states $\ket{\text{Werner}_n}$ have entanglement detection lengths $l_{Ent}(\ket{Werner_n})\in[2,n]$, while nonlocality lengths $l_{Nol}(\ket{Werner_n})=+\infty$. 
A more complicated state also exists, showing gaps between nonlocality detection lengths and GME detection lengths. Bowles $et\ al.$ introduce a genuinely entangled state that admits a fully local model, hence violating no Bell's inequality \cite{Bowles2016GME_LHV}, which implies that these quantities are not equivalent.

\begin{proposition}\label{prop:GME_LHV}
Denote the $n$-qubit state with two parameters \begin{equation}
  \rho_n^*(\alpha,\theta)=\begin{pmatrix}
    \gamma(t(0)) &0 &\cdots &(\alpha cs)^n\\
    0 &\gamma(t(1)) &\cdots &0\\
    \vdots &\vdots &\ddots &\vdots\\
    (\alpha cs)^n &0 &\cdots &\gamma(t(2^n))
  \end{pmatrix}
\end{equation}
where $c=\cos\theta,s=\sin\theta,
  \gamma(i)=\frac{c^{2(n-i)}s^{2i}}{2^n}\left((1+\alpha)^{n-i}(1-\alpha)^i+(1+\alpha)^{i}(1-\alpha)^{n-i}\right),$
and the function $t(k)$ represents the number of $1$ of the bit-string $k$ in the binary representation. With parameters $\alpha\geq1-1/N^2$ and $\theta>0$, the GME detection length satisfies $l_{GME}(\rho_n^*(\alpha,\theta))=n$. 
\end{proposition}

Based on Proposition \ref{prop:GME_LHV}, when parameters satisfy 
  $\cos^2(2\theta)\geq \frac{2\alpha-1}{(2-\alpha)\alpha^3}$,
the state $\rho_n^*(\alpha,\theta)$ violates no Bell's inequality \cite{Bowles2016GME_LHV,Bowles2016unsteer}, which means that for some parameters $\alpha,\theta$ the nonlocality detection length $l_{Nol}(\rho_n^*(\alpha,\theta))=+\infty$, while the GME detection length $l_{GME}(\rho^*_{n}(\alpha,\theta))=n$.

In Table \ref{table:det_length}, we summarize the detection length for different forms of entanglement across various quantum states. Some of these results are validated through previous propositions or ascertained by entanglement witness (EW) instances derived from our subsequent semidefinite programming (SDP) analysis, and the remaining results are proved in earlier literature.

\begin{table}[!htp]
\begin{tabular}{|c|c|c|c|c|}
\hline
Detected state $\rho$ &$l_{GME}(\rho)$ &$l_{Bipa}(\rho,\{1\})$& $l_{Ent}(\rho)$& $l_{Nol}(\rho)$\\
\hline
$\ket{GHZ_n}$&$n$~\cite{shi2024entanglement}&$n$&$n$&$n$~\cite{toth2005stabilizer}\\
\hline
$\ket{Dicke_n^k}$&$2$~\cite{shi2024entanglement}&$2$&$2$&$2$~\cite{Tura2015two_body}\\
\hline
$\ket{Cluster_n}$, $\ket{Ring_n}$&$3$~\cite{shi2024entanglement}&$3$&$3$&$3$~\cite{Guhne2005bell_graph}\\
\hline
$0.6*\ketbra{\psi_n}{\psi_n}+0.4*\ketbra{GHZ_n}{GHZ_n}$& $n$& $2$& $2$& $[2,+\infty]$ \\
\hline
$0.6*\ketbra{\psi_n'}{\psi_n'}+0.4*\ketbra{GHZ_n}{GHZ_n}$& $n$& $n$& $2$& $[2,+\infty]$ \\
\hline
$\rho_n^*$&$n$& $n$ &$n$ &$+\infty$~\cite{Bowles2016GME_LHV}\\
\hline
\end{tabular}
\captionsetup{justification=raggedright}
\caption{This table presents various detection lengths for some typical states, where $\ket{GHZ_n}=(\ket{00\cdots 0}+\ket{11\cdots 1})/\sqrt{2}$, $\ket{Dicke_n^k}=\left(\sum_{(i_1,i_2,\cdots,i_n)\in \mathbb{Z}^n_2,wt(i_1,i_2,\cdots,i_n)=k}\ket{(i_1,i_2,\cdots,i_n)}\right)/\sqrt{\binom{n}{k}}$, $\ket{Cluster_n}$ and $\ket{Ring_n}$ are graph states corresponding to the ``line" and ``ring" graph with $n$ vertices, $\ket{\psi_n}=(\ket{100\cdots 0}+\ket{010\cdots 0})/\sqrt{2}$, $\ket{\psi_n'}=(\ket{0\cdots 001}+\ket{0\cdots 010})/\sqrt{2}$, and $\rho_n^*$ is the state defined in Proposition \ref{prop:GME_LHV} with $\alpha=1-1/n^2$ and $\theta=\pi/4$.}
\label{table:det_length}
\end{table}

Utilizing these formulations, we extend the definition of detection length to two entanglement configurations: entanglement intactness \cite{Horodecki2009intact} and entanglement depth \cite{Sorensen2001depth, Toth2005depth}. For the sake of brevity, the detailed outcomes of this analysis are delineated in Appendix \ref{appen:ent_dep}.

\section{Numerical construction and analyses}\label{sec:SDP}
\subsection{Detecting Entanglement and Nonlocality via SDP}
Knowing the information about detection lengths for various states, it is essential to construct an entanglement witness (EW) that precisely corresponds to these lengths for practical implementation.
To enhance the efficiency of entanglement detection, 
we apply a numerical search algorithm, specifically SDP, to construct EWs with limited length. It is noteworthy that this numerical technique can also be employed to estimate the detection length for any given state: constructing an EW with a limited detection length for a given state, or formulating a Bell's inequality with a limited length for correlators, effectively serves as an upper bound on the detection length for our analyses.

Given an $n$-qubit state $\rho\in \cD$, and a collection of subsets $\mathcal{M}=\{M_1,M_2,\cdots,M_k\}$, with $M_j\subset [n]$, we execute the following SDP: \begin{equation}\label{SDP:ent}
  \begin{split}
    &\min \tr(W\rho)\\
    &\text{s.t. }\tr(W)=d, W=\sum_{M_j\in\mathcal{M}}H_{M_j}\otimes\mathbb{I}_{\overline{M_j}},\\
  &\ \ \ \ \ \ W=P+Q^{T_S}\ \ \ \exists \ S\subset [n],P\geq0,Q\geq0,
  \end{split}
\end{equation}
where $d=2^n$ denotes the system size, $T_S$ is the partial transpose over $\cH_S$, and
$H_{M_j}$ is a Hermitian operator on $\cH_{M_j}$.
The restriction $\tr(W)=d$ is simply to facilitate a unified normalization. The operator $W=P+Q^{T_S}$ is called \emph{decomposable witness} \cite{Lewenstein2000witness}. This EW has a positive expectation
value on every state that has positive
 partial transpose (PPT) over $\cH_S$.  
If the minimum of the SDP for state $\rho$ is negative, a condition we henceforth denote as ``$\rho$ passing the SDP test," then there exists a subset $S$ such that $\rho$ is not a PPT state over $\cH_{S}$, and is consequently identified as an entangled state on $S|\overline{S}$. For a given EW $W$ (or an operator, in general), we define the length of $W$ as the maximal number of parties involved simultaneously. Then the EWs derived from the SDP associated with $\cM$ are all of $\max(\cM)$-$length$. By assigning the marginals set $\cM$, we can restrict the length of witness return by the SDP, and then construct EW for given states with the exact detection length.
Note that, constrained by the decomposable witness, the SDP test does not certify all entangled states, specifically those classified as PPT entangled states. Moreover, although the SDP cannot precisely ascertain the entanglement detection length for all states, it can provide an upper bound for those states that pass the SDP test.
With the same technique, we also introduce SDP to construct a witness for GME  or bipartitioned entanglement with limited length, as detailed in Appendix \ref{appdix:SDP}.

For a given state $\rho$, the nonlocality detection length $l_{Nol}(\rho)$ is upper-bounded by existing formulations of Bell's inequalities, while the entanglement detection lengths of this state $\l_{Ent}(\rho)$, as claimed in Eq.~\eqref{eq:nonlocality}, serve as a lower bound. Consequently, we could determine the nonlocality detection lengths for several classes of state, as presented in Table \ref{table:det_length}.
To further formulate Bell's inequality with the precise nonlocality detection length $l_{Nol}$, we impose constraints on the preceding SDP and empirically transform the resulting EW into a Bell's inequality \cite{toth2005stabilizer, zhao2022bell,Baccari2020bell,Wu2023bell}. 
For a given target state $\rho$ and marginal set $\mathcal{M}$, the constraint version of SDP reads: \begin{equation}\label{SDP:nonlocal}
  \begin{split}
    &\min\ \tr(W\rho)\\
    &\exists\ S\subset [n],P\geq0,Q\geq0,\alpha_i\in\mathbb{C}\\
    &\text{s.t. }W=P+Q^{T_S},\tr(W)=d,\\ 
  &\ \ \ \ \ W=\sum_{M_j\in\mathcal{M}}\sum_i\alpha_i\left(\bigotimes_{m\in M_j}O_{m,i}\right) \otimes\mathbb{I}_{\overline{M_j}},
  \end{split}
\end{equation}
where each observer $O_{m,i}$ is either a $Z$, $X$ or $I$ operator on the $m$-th qubit. 
Note that this SDP would return an entanglement witness $W$ composed of terms that are tensor products of solely the $X$, $Z$, and $I$ operators. Then we perform the following assignment on observers \cite{Makuta2021self_test_stab, Santos2023selftest}: 
\begin{equation}
\begin{split}
  X_i&\rightarrow\begin{cases}
    (A_i^{(0)}+A_i^{(1)})/\sqrt{2},&\text{if }i=1\\
    A_i^{(0)},&\text{otherwise}
  \end{cases},\\
  Z_i&\rightarrow\begin{cases}
    (A_i^{(0)}-A_i^{(1)})/\sqrt{2},&\text{if }i=1\\
    A_i^{(1)},&\text{otherwise}
  \end{cases}.
\end{split}
\end{equation}
By this assignment, we transform the obtained EW $W$ into a Bell's inequality $\cB$ with the same length $max(\cM)$. However, the classical limit of this constructed Bell's inequality may not be violated by the target state $\rho$. This Bell's inequality is deemed nontrivial only when the expectation value $\tr(W\rho)$ exceeds the corresponding classical limit. 


\subsection{Entanglement Detection in Noisy Environments}\label{sec:noise}
When experimentally detecting entanglement, practical measurement implementations are not always precise, including both the procedures of state preparation and measurement. Yet, the entanglement or GME witness is somehow robust against noise for the target state. This robustness ensures the EW's availability in practice and also serves as a measure to evaluate an EW \cite{toth2005stabilizer}.

Let $W$ be a witness capable of detecting specific entanglement forms from $\rho$ satisfying $\tr(W)=d$. Assume that, due to noise in the experimental environment, the actual measured state $\rho'$ is mixed with the global depolarizing noise \begin{equation}\label{eq:noise_form}
  \begin{split}
    \rho'=\Lambda_{dep}(\rho)=(1-p)\rho+\frac{p\mathbb
    {I}}{d}.
  \end{split}
\end{equation}
Let $\Tilde{p}(W,\rho):=\frac{-\tr(W\rho)}{1-\tr(W\rho)}$ denote the noise tolerance parameter for a given EW $W$.
When the detected state is mixed with noise $p<\Tilde{p}(W,\rho)$, the expectation value of the detected state for witness $W$ is still negative, 
thereby $\rho'$ is still detected as entangled by $W$. For some typical states, we summarize the noise tolerance parameters of witnesses obtained from our SDP corresponding to different forms of entanglement.
We also assess the noise tolerance parameter for these witnesses under local depolarizing noise, yielding similar results to those observed with global depolarizing noise. The noise tolerance outcomes and detailed analyses for both scenarios are detailed in Appendix \ref{appdix:local}.

Through simulation, we conclude a trade-off between noise robustness and detection length for entanglement witnesses: EWs with greater length are challenging to implement experimentally, yet they offer increased noise resilience.
For instance, a $2$-length EW for the $3$-qubit W state $\ket{W_3}$ tolerates noise up to $\Tilde{p}=0.5101$ at most, whereas a $3$-length EW could accommodate noise up to $\Tilde{p}=0.7904$.
This trade-off feature also holds for other entanglement forms, like GME. Moreover, the number of distinct marginal settings $|\cM|$ within an EW $W(\cM)$ affects its noise robustness: an EW that includes more diverse marginal settings tends to exhibit greater noise robustness. For example, a GME witness for the $3$-qubit W state $\ket{W_3}$, limited to 2-local measurements on the marginal settings $\{12\}$ and \{23\}, tolerates noise up to $\Tilde{p}=0.1859$. In comparison, a GME witness incorporating $2$-local measurements across all marginal settings, \{\{12\}, \{23\}, \{13\}\}, shows enhanced noise tolerance $\Tilde{p}=0.3039$. 

In addition to noise inherent in state preparation and storage, measurement implementation can also induce noise, particularly for longer-length observables that are more susceptible to noise.
Let us consider a typical case that $bit$-$flip$ $errors$ occur with a uniform probability $\epsilon$ on each qubit during the measurement process. For example, during a $Z$-basis measurement, a qubit that initially collapses to the state $\ket{0}$ may, with probability $\epsilon$, randomly flip to $\ket{1}$, thereby altering the measurement outcome from $+1$ to $-1$. 
When incorporating the preceding global depolarizing noise channel along with potential bit-flip errors in measurements, we could derive the final expectation value obtained in practice for EW $W$. 

\begin{proposition}
Under the experimental environment of global depolarizing channel with parameter $p$ and measurement bit-flip error with probability $\epsilon$, the expectation value of a $k$-length witness $W$, which satisfies $\tr(W)=d$, is given by
    \begin{equation}
\begin{split}
  \alpha^*(\rho)&=1-(1-2\epsilon)^k(1-(1-p)\tr(W\rho)-p),
\end{split}
\end{equation}
indicating that the noise tolerance parameter of $W$ witness for $\rho$ is determined as \begin{equation}
    p^*:=1-1/\left((1-2\epsilon)^k(1-\tr(W\rho))\right).
\end{equation} 
\end{proposition} 

The proof of this proposition can be found in Appendix \ref{appdix:proof}. In practical experiments nowadays, we could implement measurement with small error probability $\epsilon$ to collect useful information. Suppose $\epsilon$ is small enough to ensure $W$ can still detect entanglement from $\rho$ \begin{equation}\label{eq:eps_req}
  \overline{\alpha}(\rho)<0\Rightarrow \epsilon<\frac{1}{2}\left(1-(1-\tr(W\rho))^{-1/k}\right).
\end{equation}

\begin{figure}[!h]
    \centering
    \subfigure[\ $\epsilon=0.08$]{ 
            \label{} 
            \includegraphics[width=5.8cm]{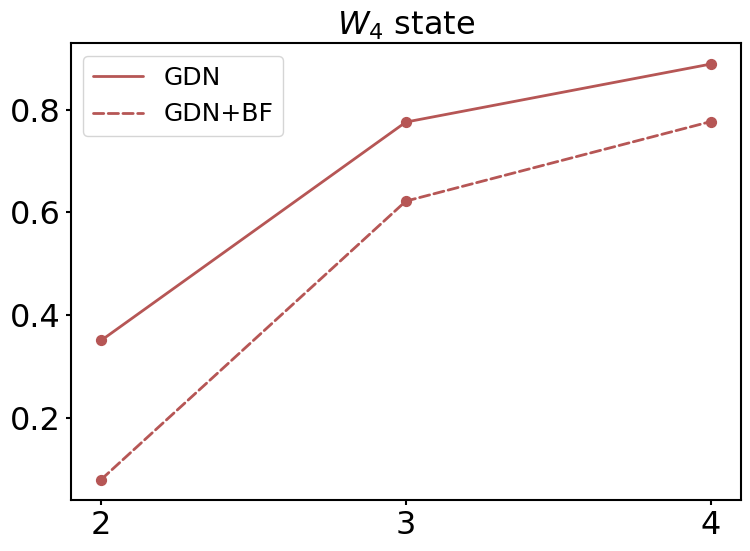}}
    \subfigure[\ $\epsilon=0.08$]{ 
            \label{} 
            \includegraphics[width=5.8cm]{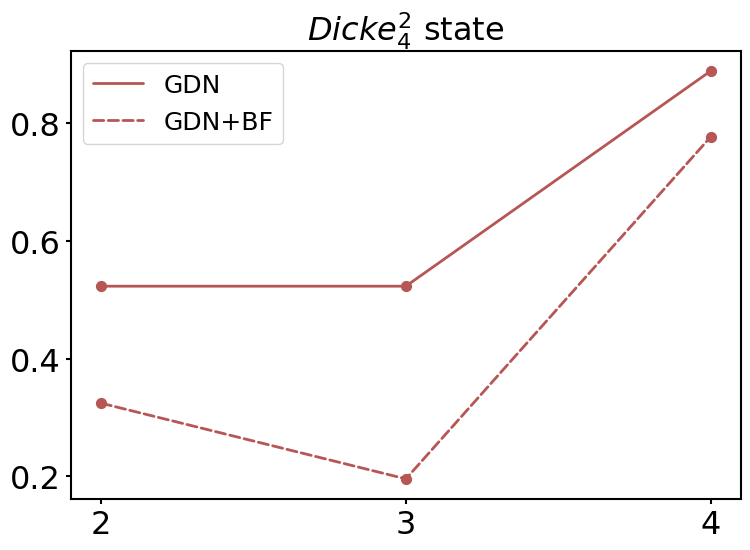}} 
    \subfigure[\ $\epsilon=0.205$]{ 
            \label{} 
            \includegraphics[width=5.8cm]{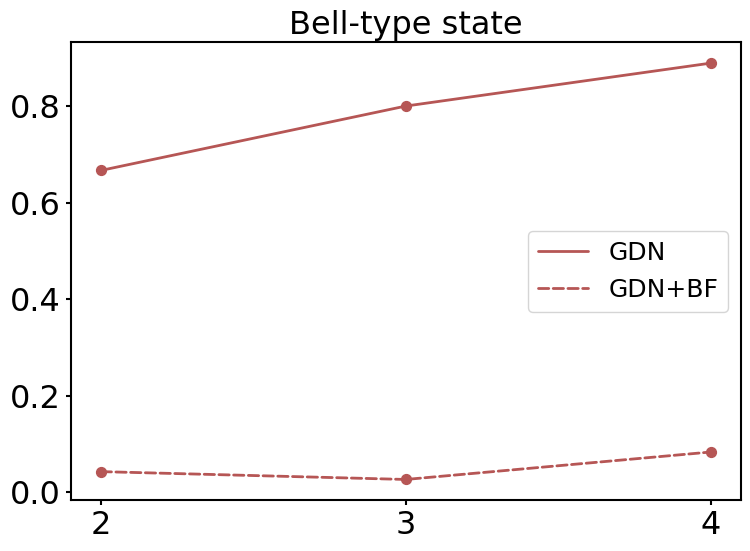}} 
\captionsetup{justification=raggedright}
    \caption{Noise tolerance for different witness lengths. Through numerical simulation, we compare the entanglement noise tolerance parameter $\Tilde{p}_{Ent}$ for various types of states, under the global depolarizing noise (GDN) or plus the bit-flip error (GDN+BF), where 
    the bell-type state is $\rho_2=\ket{\psi^-}\otimes\ket{00}$. When the bit-flip error probability achieves $\epsilon=0.08$ ($\epsilon=0.205$), the $2$-length witness for the Dicke state $\ket{Dicke_4^2}$ (Bell-type state $\rho_2$) tolerates higher noise than the $3$-length witness.}
    \label{fig:GDN+BF}
\end{figure}

Utilizing the noise tolerance $p^*$ as a measure for EW robustness, we conclude that even though EW with longer length may possess stronger detection ability, the bit-flip measurement error would incur robustness decreasing. Consider two entanglement witnesses for the $N$-qubit cluster state with distinct detection length: \begin{equation}
  \begin{split}
    W_1&=\frac{1}{d(N-1)}\sum_{1\leq i\leq N-1}\left(\mathbb{I}-\overline{S}_i-\overline{S}_{i+1}\right),\\
    W_2&=\frac{1}{d}\prod_{0\leq i\leq K-1}\left(\mathbb{I}-\overline{S}_{4i+1}-\overline{S}_{4i+2}-\overline{S}_{4i+1}\overline{S}_{4i+2}\right),
  \end{split}
\end{equation}
where $\overline{S}_i$ denotes the stabilizer for the cluster state \begin{equation}
  \begin{split}
  &\overline{S}_1:=X_1Z_2,\ \overline{S}_N:=Z_{N-1}X_N,\\
    &\overline{S}_k:=Z_{k-1}X_kZ_{k+1},\ k=2,3,\ldots,N-1,
  \end{split}
\end{equation} 
and $K=\lfloor(N+2)/4\rfloor$. The first EW has detection length $3$, tolerating noise when $p<1/2$. When implementing the second EW, we globally conduct two measurement settings, thereby it has length $\min(4K-1,N)$, and it tolerates noise for Cluster state when $p<2^K/(2^K+1)$. When the bit-flip error occurs with probability $\epsilon$, the noise tolerances of these two witnesses would decrease to $p^*(W_1)=1-1/2(1-2\epsilon)^3$, $p^*(W_2)=1-1/\left((1+2^K)(1-2\epsilon)^{\min(4K-1,N)}\right)$. Combined with the requirement stated in Eq.~\eqref{eq:eps_req}, we find that the shorter EW $W_1$ possess higher (non-zero) noise tolerance than the longer $W_2$, i.e. $p^*(W_1)>p^*(W_2)>0$, when \begin{equation}\label{eq:eps_req2}
  \left(1-\sqrt[C]{\frac{2}{1+2^K}}\right)/2<\epsilon<\left(1-\max\left(\sqrt[C+3]{\frac{1}{1+2^K}},\sqrt[3]{\frac{1}{2}}\right)\right)/2,
\end{equation}
where \begin{equation}
  C=\begin{cases}
    4K-5,&N\bmod{4}=2\\
    4K-4,&\text{otherwise}
  \end{cases}.
\end{equation}

The $\epsilon$ value obeyed the requirement in Eq.~\eqref{eq:eps_req2} indeed exists in some instances. For example, consider the $11$-qubit cluster state, when $0.0857<\epsilon<0.0905$, we have $p^*(W_1)>p^*(W_2)>0$.

For certain states, we have compared the entanglement noise tolerance parameter $\Tilde{p}_{Ent}$ under bit-flip error conditions and without such errors for EWs of varying detection lengths, as illustrated in Fig.~\ref{fig:GDN+BF}. Additionally, we have examined the noise tolerance for GME, which exhibits asymptotic similarity to the aforementioned findings; further details are omitted for brevity.

\subsection{Practical Applications}\label{subsec:app}
Our method's direct applicability to any quantum state makes it suitable for a broad range of entangled states in experimental settings, including mixed states, unlike previous works focused on pure states. By employing SDP, we can construct entanglement witnesses for arbitrary detection lengths, as detailed in Appendix \ref{appdix:construct}. Furthermore, SDP can assess detection length properties for specific quantum states that are hard to analyze theoretically.

Let us consider a kind of entangled state which is significant and easy in experimental contexts \cite{zhang2021ent_gene, Zhou2022exp}: \begin{equation}\label{eq:exp_state}
  \ket{\Phi}=\bigotimes_{i=1}^N\sqrt{\text{SWAP}}_{\{2i,2i+1\}}\ket{\Psi^-}^{\otimes N},
\end{equation}
where $\ket{\Psi^-}=(\ket{01}-\ket{10})/\sqrt{2}$ denotes a Bell's state.
Through numerical simulation, we determine both the entanglement and GME detection length $l_{Ent}(\ket{\Phi})=l_{GME}(\ket{\Phi})=2$ for $N=3$ cases, with noise tolerance parameter $p_{Ent}(\ket{\Psi})=0.9697$ and $p_{Ent}(\ket{\Psi})=0.5908$. The obtained GME witness possesses higher noise tolerance than the original one for this state \cite{zhang2021ent_gene}. 

Furthermore, we develop several instances of nontrivial Bell's inequalities through numerical simulation, which are maximally violated by some graph states and achieve their corresponding nonlocality detection length $l_{Nol}(\rho)$. Note that these outcomes for graph states are anticipated, given that prior research has formulated Bell's inequalities for such states based on their stabilizer properties \cite{toth2005stabilizer,Guhne2005bell_graph,Wu2023bell,Baccari2020bell}. However, our approach potentially yields Bell's inequalities with a more favorable ratio between the quantum and classical bounds, thereby leading to higher noise robustness. For example, our SDP method proposes Bell's inequality for the $4$-qubit Cluster state and Ring state respectively 
\begin{align}
    \cB\left(\ket{Cluster_4}\right)&:\langle(A_1^{(0)}+A_1^{(1)})(A_2^{(1)}+A_3^{(0)}A_4^{(1)})\rangle+\langle (A_1^{(0)}-A_1^{(1)})(A_2^{(0)}A_3^{(1)}+A_2^{(0)}A_4^{(0)})\rangle\leq \beta_C=4,\ \beta_Q=4\sqrt{2},\\
\cB\left(\ket{Ring_4}\right)&:\langle (A_1^{(0)}+A_1^{(1)})(A_3^{(0)}+A_2^{(1)}A_4^{(1)})+(A_1^{(0)}-A_1^{(1)})A_3^{(1)}(A_2^{(0)}+A_4^{(0)})\rangle\leq\beta_C=4,\ \beta_Q=4\sqrt{2}.
  \end{align}
The maximal quantum value $\beta_Q=$ of both can be reached by taking $A^{(0)}_1=(X_1+Z_1)/\sqrt{2}$, $A^{(1)}_1=(X_1-Z_1)/\sqrt{2}$ and $A^{(0)}_i=X_i$, $A^{(1)}_i=Z_i$ for other $i$. Denote the ratio of the classical bound to the quantum bound of Bell's inequality by $r:=\beta_C/\beta_Q$. Note that, when mixed with global depolarizing noise $p<1-r$, the detected state $\rho'=(1-p)\rho+\frac{p\mathbb{I}}{d}$ still violates the Bell's inequality. For these two Bell's inequalities, the bound ratios $r(\ket{Cluster_4})=r(\ket{Ring_4})=1/\sqrt{2}$ are smaller than those of the naive Bell's inequality constructions for these two states $r_0(\ket{Cluster_4})=r_0(\ket{Ring_4})=3/4$ \cite{Guhne2005bell_graph}. 
Therefore, with a smaller bound ratio $r$, our derived Bell's inequalities tolerate more noise. 
This is attributed to our SDP methodology, which is designed to identify EW that maximizes the resilience to depolarizing noise. However, when a Bell's inequality contains global operators, the ratio can be minimized to $r(\ket{Ring_4})=1/2$, which is the proven theoretical limit for this state \cite{Toth2006bell_graph}. Although our SDP-derived Bell's inequality for the $\ket{Ring_4}$ state exhibits lower noise robustness, it only contains $3$-length operators, offering a tradeoff alternative. Additionally, other literature has also presented Bell's inequalities employing only $3$-length operators for the $\ket{Ring_4}$ state, possessing the same level of noise robustness \cite{Baccari2020bell,Wu2023bell}.


Our SDP method also efficiently addresses many other scenarios: it ascertains the entanglement detection length of $2$-dimension multipartite Werner state, which is entangled if and only if it violates the PPT separability criterion \cite{Eggeling2000Werner}; for the state $\rho=p\ketbra{\psi_n}{\psi_n}+(1-p)\ketbra{{\ghz}_n}{{\ghz}_n}$ defined in Proposition \ref{prop:ent_GME_gap}, we offer an analytic formulation for a $2$-length EW capable of detecting the bipartitioned entanglement; We additionally establish an upper bound on the detection length for symmetric states for any given parameter, and construct EW with corresponding detection length. For the sake of simplicity, we present the technical details in the Appendix \ref{appdix:prac}.

\section{Conclusion and Outlook}
The concept of ``detection length''  constitutes a metric for evaluating the minimal observable requirement necessary for entanglement detection. In this work, we extended the analytical framework to encompass the detection length for diverse entanglement forms and nonlocality, which is readily applicable to analyzing any given entanglement scenario. Furthermore, we systematically employed the SDP approach to construct entanglement witnesses with the desired detection length. Utilizing this SDP methodology, numerous practical scenarios concerning detection length are effectively addressed. Additionally, we conducted a numerical analysis to assess the noise robustness of witnesses across various states, revealing a trade-off between the noise robustness and the detection length. 
In analyzing the bit-flip measurement error, we compared noise robustness between a long-length EW and a short-length EW derived from our SDP for the Cluster state. Our calculation indicated that the EW with a shorter detection length exhibits better performance when the bit-flip error probability satisfies certain criteria. 
For future exploration, it is worthwhile to experimentally implement and analyze the performance of entanglement witnesses with limited detection lengths. Moreover, the SDP method for constructing witnesses faces scalability challenges. The exponential growth of SDP coefficients with the size of the quantum system makes it hard to construct EWs for large-scale quantum systems efficiently.

\begin{acknowledgements}
We express our gratitude to Lin Chen, Chu Zhao, and Jue Xu for their insightful discussions. This work acknowledges funding from HKU Seed Fund for Basic
 Research for New Staff via Project 2201100596, Guangdong Natural Science Fund—General Programme via Project 2023A1515012185, National Natural Science Foundation of China (NSFC) Young Scientists Fund via Project 12305030, 27300823, Hong Kong Research Grant Council (RGC) via No. 27300823, and NSFC/RGC Joint Research Scheme via Project N\_HKU718/23.
\end{acknowledgements}

\appendix

\section{Technical proof}\label{appdix:proof}
\setcounter{theorem}{1} 
\renewcommand{\thetheorem}{\arabic{theorem}}
\setcounter{proposition}{1} 
\renewcommand{\theproposition}{\arabic{proposition}}
\setcounter{lemma}{0} 
\renewcommand{\thelemma}{\arabic{lemma}}

In this section, we present the details of technical proof for several propositions and lemma. First, we provide a simple observation.
\begin{observation}\label{observation : rho}
  For an entangled state $\rho$, if there exists a subset $S$, such that the marginal state $\rho_S$ is also an entangled state, then $l_{Ent}(\rho)\leq |S|$. Specially, if $|S|=2$, then $l_{Ent}(\rho)=2$. 
\end{observation}
\begin{proof}
    If $\rho_S$ is entangled, then the compatibility set $\cC(\rho,\{S\})$ contains only entangled states.  This means that  $\{S\}$ detects $\rho$'s entanglement, and $l_{Ent}(\rho)\leq |S|$. 
\end{proof}


\setcounter{proposition}{0}
\begin{proposition}
Let $\ket{\psi_n}=\frac{1}{\sqrt{2}}(\ket{100\cdots 0}+\ket{010\cdots 0})$, $\ket{\ghz_n}=\frac{1}{\sqrt{2}}(\ket{0\cdots 00}+\ket{1\cdots 11})$ and 
\begin{equation*}
\rho=p\ketbra{\psi_n}{\psi_n}+(1-p)\ketbra{{\ghz}_n}{{\ghz}_n},  
\end{equation*}
where $\frac{1}{2}<p<1$, then we have $l_{GME}(\rho)=n>l_{Ent}(\rho)=2$. And for any bipartition $G\subset [n]$, if $G$ contains only one of $1$ and $2$, $l_{Bipa}(\rho,G)=2$; otherwise $l_{Bipa}(\rho,G)=n$.
\end{proposition}
\begin{proof}
When $0<p<1$, the range $\cR(\rho)$ is spanned by $\{\ket{\psi_n}, \ket{\ghz_n}\}$, and $\dim(\cR(\rho))=2$. Since all the biproduct states in $\cR(\rho)$ are $\{a\ket{\psi_n}\mid |a|=1, a\in \bbC \}$, then $\rho$ is a genuinely entangled state. Since
 $\cC(\rho, \cS_{n-1})$ contains a biseparable state $\sigma=p\ketbra{\psi_n}{\psi_n}+\frac{1-p}{2}\ketbra{0\cdots 00}{0\cdots 00}+\frac{1-p}{2}\ketbra{1\cdots 11}{1\cdots 11}$, we have  $l_{GME}(\rho)=n$.
  Note that $\rho_{\{1,2\}}=\frac{p}{2}(\ket{01}+\ket{10})(\bra{01}+\bra{10})+\frac{1-p}{2}\ketbra{00}{00}+\frac{1-p}{2}\ketbra{11}{11}$. When $\frac{1}{2}<p<1$, $\rho_{\{1,2\}}$ is NPT, and $\rho_{\{1,2\}}$ is an entangled state. By Observation \ref{observation : rho}, we have $l_{Ent}(\rho)=2$. Similarly, since the state $\sigma$ contained in $\cC(\rho,\cS_{n-1})$ is biseparable under bipartition $G$ if $G$ contains both particles $1$ and $2$ or neither (these two situations are equivalent), we have $l_{Bipa}(\rho,G)=n$. In the other hand, since $\rho_{\{1,2\}}=\frac{p}{2}(\ket{01}+\ket{10})(\bra{01}+\bra{10})+\frac{1-p}{2}\ketbra{00}{00}+\frac{1-p}{2}\ketbra{11}{11}$ is NPT when $\frac{1}{2}<p<1$, $\rho$ is entangled under bipartition $G$ in this case if $G$ contain exactly one of $1$ and $2$, we have $l_{Bipa}(\rho,G)=2$ under this bipartition $G$.
\end{proof}

\begin{proposition}
Denote the $n$-qubit state with two parameters \begin{equation}
  \rho_n^*(\alpha,\theta)=\begin{pmatrix}
    \gamma(t(0)) &0 &\cdots &(\alpha cs)^n\\
    0 &\gamma(t(1)) &\cdots &0\\
    \vdots &\vdots &\ddots &\vdots\\
    (\alpha cs)^n &0 &\cdots &\gamma(t(2^n))
  \end{pmatrix}
\end{equation}
where $c=\cos\theta,s=\sin\theta,
  \gamma(i)=\frac{c^{2(n-i)}s^{2i}}{2^n}\left((1+\alpha)^{n-i}(1-\alpha)^i+(1+\alpha)^{i}(1-\alpha)^{n-i}\right),$
and the function $t(k)$ represents the number of $1$ of the bit-string $k$ in the binary form. With parameters $\alpha\geq1-1/N^2$ and $\theta>0$, the GME detection length satisfies $l_{GME}(\rho_n^*(\alpha,\theta))=n$. 
\end{proposition}
\begin{proof}
  With $\alpha\geq1-1/N^2$ and $\theta>0$, the entanglement concurrence of this state satisfies \cite{Bowles2016GME_LHV} \begin{equation}
    C(\rho^*_{n}(\alpha,\theta):=\frac{2\sin^n(2\theta)\left(\alpha^n+\frac{1+\alpha}{2}^n+\frac{1-\alpha}{2}^n-1\right)}{(1+\alpha\cos 2\theta)^n+(1-\alpha\cos 2\theta)^n}>0,
  \end{equation}
  which means $\rho^*_{n}(\alpha,\theta)$ is genuinely entangled state. Since $\cC(\rho^*_{n}(\alpha,\theta), \cS_{n-1})$ contains a fully separable state \begin{equation}
    \sigma=\begin{pmatrix}
    \gamma(t(0)) &0 &\cdots &0 \\
    0 &\gamma(t(1)) &\cdots &0\\
    \vdots &\vdots &\ddots &\vdots\\
    0 &0 &\cdots &\gamma(t(2^n))
  \end{pmatrix},
  \end{equation}
   the state $\rho^*_{n}(\alpha,\theta)$ with parameters $\alpha\geq1-1/N^2$ and $\theta>0$ has GME detection length $l_{GME}(\rho^*_{n}(\alpha,\theta))=n$.
\end{proof}

\begin{proposition}
Under the experimental environment of global depolarizing channel with parameter $p$ and measurement bit-flip error with probability $\epsilon$, the expectation value of a $k$-length witness $W$, which satisfies $\tr(W)=d$, is given by
    \begin{equation}
\begin{split}
  \alpha^*(\rho)&=1-(1-2\epsilon)^k(1-(1-p)\tr(W\rho)-p),
\end{split}
\end{equation}
indicating that the noise tolerance parameter of $W$ witness for $\rho$ is determined as \begin{equation}
    p^*:=1-1/\left((1-2\epsilon)^k(1-\tr(W\rho))\right).
\end{equation} 
\end{proposition} 
\begin{proof}
    For an EW with preceding normalization restriction $\tr(W)=d$, it can be written as $W=\mathbb{I}-\sum_i P_i$, where $P_i$ is a non-identity $k$-length Pauli operator. When we are deriving the expactation value of $W$, we only measure these $P_i$ terms. Consequently, if we want to analyze the influence of bit-flip error on measurement outcome, we only need to consider these $P_i$ terms. With bit-flip error of probability $\epsilon$, each $k$-length measurement would read out an opposite outcome with the probability of $w:=\sum_{i<k,i\bmod2= 1}\binom{k}{i}\epsilon^{i}(1-\epsilon)^{k-i}$. Assume that the probability of $+1$ result is $a$ and $-1$ as $(1-a$), indicating the original expectation value equal to $a-(1-a)=2a-1$. With the bit-flip error on measurement, the expectation value for a $k$-length measurement is modified to:
\begin{equation}
  a\cdot(1-w)+a\cdot (-1)\cdot w+(1-a)\cdot(-1)\cdot(1-w)+(1-a)\cdot w=(1-2w)(2a-1),
\end{equation}
can be seen as shrinking a coefficient \begin{equation}
\begin{split}
  1-2w&=(1-\epsilon+\epsilon)^k-2\sum_{\text{odd }i}\binom{k}{i}\epsilon^{i}(1-\epsilon)^{k-i}\\
  &=\sum_{\text{even }i}\binom{k}{i}\epsilon^{i}(1-\epsilon)^{k-i}-\sum_{\text{odd }i}\binom{k}{i}\epsilon^{i}(1-\epsilon)^{k-i}\\
  &=\sum_{\text{even }i}\binom{k}{i}(-\epsilon)^{i}(1-\epsilon)^{k-i}+\sum_{\text{odd }i}\binom{k}{i}(-\epsilon)^{i}(1-\epsilon)^{k-i}\\
  &=(1-2\epsilon)^k.
\end{split}
\end{equation}
    Hence, the expectation value for a $k$-length EW is given by \begin{equation}\label{alphaBar}
      \begin{split}
      \overline{\alpha}(\rho)&=\tr(\rho)-(1-2\epsilon)^k\tr(P\rho)\\
      &=1-(1-2\epsilon)^k(1-\tr(W\rho)),
      \end{split}
    \end{equation}
    under solely the local bit-flip measurement error. Then substituting the expectation value of the state mixed with global depolarizing noise in Eq.~\eqref{eq:noise_form} into this form, we could complete the calculation of $\alpha^*$. 
\end{proof}

\section{Detection length for entanglement depth and intactness}\label{appen:ent_dep}
Based on the analytical model proposed previously, we can further analyze the proposition for the detection length for entanglement depth and entanglement intactness.
\subsection{Entanglement depth}
For a pure state $\ket{\Psi}$ which is the tensor product of multi-particle quantum states \begin{equation}
  \ket{\Psi}=\ket{\phi_1}\otimes\ket{\phi_2}\otimes\cdots\otimes\ket{\phi_n},
\end{equation}
it is said to be $k$-producible if all $\ket{\phi_i}$ are states of at most $k$ particles. A mixed state is called $k$-producible if it is a convex combination of pure states that are all at most $k$-producible. 
We say a quantum state has an entanglement depth $k$, if it is $k$-producible but not $(k-1)$-producible. 
Based on this definition, we say that a state $\rho$ has $\mathcal{S}$-detectable $k$-depth entanglement if the compatibility set $\mathcal{C}(\rho,\mathcal{S})$ contains only states having $k'$-depth entanglement with $k'\geq k$, i.e. at least having $k$-depth entanglement. When in this case, we say that $\mathcal{S}$ detects $\rho$'s $k$-depth entanglement.
Similar to GME detection length, we define the $k$-depth entanglement detection length of state $\rho$ as \begin{equation}
  l_{dep}^k(\rho)=\min_{\mathcal{S}}\{\max(\cS) \mid \text{$\mathcal{S}$ detects $\rho$'s $k$-depth entanglement}\}.
\end{equation}
And the minimum number of $l^k_{dep}$-body marginals needed to detect $k$-depth entanglement $m^k_{dep}$ as \begin{equation}
  m^k_{dep}(\rho)=\min_{\mathcal{S}:|S|=l_{dep}^k(\rho),\forall S\in\mathcal{S}}\{|\mathcal{S}|\mid\mathcal{S}\text{ detects $\rho$'s $k$-depth entanglement}\}.
\end{equation}
\begin{proposition}
For a pure state $\ket{\Psi}$ with $k$-depth entanglement $\ket{\Psi}=\ket{\phi_1}\otimes\ket{\phi_2}\otimes\cdots\otimes\ket{\phi_n}$, let the $D=\{\ket{\phi_i}\mid |\phi_i|=k\}$ denotes the set of $k$ particles inseparable states, the values of the $k$-depth entanglement detection length $l_{dep}^k(\ket{\Psi})$ and the corresponding minimum $m_{dep}^k(\ket{\Psi})$ satisfy \begin{equation}
  \begin{split}
    l_{dep}^k(\ket{\Psi})&=\min_{\ket{\phi_i}\in D}l_{GME}(\ket{\phi_i})\in[2,k],\\
    m_{dep}^k(\ket{\Psi})&=m_{GME}(\arg\min_{\ket{\phi_i}\in D}l_{GME}(\ket{\phi_i})).
  \end{split}
\end{equation}
\end{proposition}
Note that these values are directly correlated to the structure of the $k$-depth entangled part of $\ket{\Psi}$. For example if one part of $\ket{\Psi}$ is $k$-qubit $W$ state, then $l_{dep}^k(\ket{\Psi})=l_{GME}(\ket{W_k})=2$. Therefore, this proposition is not applicable to $k$-depth mixed states.

\subsection{Entanglement intactness}
A pure state $\ket{\Psi}$ which can be written as the tensor product of $s$ pure state \begin{equation}
  \ket{\Psi}=\ket{\phi_1}\otimes\ket{\phi_2}\otimes\cdots\otimes\ket{\phi_s}
\end{equation}
is said to be $s$-separable. A mixed state is called $s$-separable if it is a convex combination of pure states that are all at least $s$-separable. We say a quantum state has an entanglement intactness of $s$ if it is $s$-separable but not $(s+1)$-separable. 
Like the $k$-depth entanglement part, we say that $\mathcal{S}$ detects $\rho$'s $s$-intactness entanglement if the compatibility set $\mathcal{C}(\rho,\mathcal{S})$ contains only states having $s$-intactness entanglement. Similarly define the $s$-intactness entanglement detection length and the minimum number of the corresponding marginals of state $\rho$ as \begin{equation}
  \begin{split}
  l_{int}^s(\rho)&=\min_{\mathcal{S}}\{\max_{S\in\mathcal{S}}|S| \mid \text{$\mathcal{S}$ detects $\rho$'s $s$-intactness entanglement}\},\\
  m^s_{int}(\rho)&=\min_{\mathcal{S}:|S|=l_{int}^s(\rho),\forall S\in\mathcal{S}}\{|\mathcal{S}|\mid\mathcal{S}\text{ detects $\rho$'s $s$-intactness entanglement}\}.
  \end{split}
\end{equation}

\begin{proposition}
For a pure state $\ket{\Psi}$ with $s$-intactness entanglement $\ket{\Psi}=\ket{\phi_1}\otimes\ket{\phi_2}\otimes\cdots\otimes\ket{\phi_s}$, the values of the $s$-intactness entanglement detection length $l_{int}^s(\ket{\Psi})$ and the corresponding minimum $m_{int}^s(\ket{\Psi})$ satisfy \begin{equation}
  \begin{split}
    l_{int}^s(\ket{\Psi})&=\max_i l_{GME}(\ket{\phi_i}),\\
    m_{int}^s(\ket{\Psi})&\in\left[\sum_i\lceil (|\psi_i|-1)/(l_{int}^s-1)\rceil,\sum_im_{GME}(\ket{\phi_i})\right].
  \end{split}
\end{equation}
\end{proposition}

\section{Entanglement witnesses construction via SDP}\label{appdix:SDP}

In the preceding text, we have advocated the use of SDP for constructing witnesses to detect entanglement and nonlocality, and have also demonstrated the adaptability of these methods for the detection of genuine multipartite entanglement (GME) or bipartitioned entanglement.

To detect genuine multipartite entanglement (GME) from a state $\rho$, given a marginal set $\mathcal{M}$, we can execute a similar SDP, which exploits the \emph{fully decomposable witness} to relax the constraints on GME detection \cite{jungnitsch2011taming,shi2024entanglement} \begin{equation}\label{SDP:GME}
  \begin{split}
    &\min \tr(W\rho)\\
    &\forall \ S\subset [n]\\
    &\exists \ P_S\geq0,Q_S\geq0\\
    &\text{s.t. }W=P_S+Q_S^{T_S},\tr(W)=d,\\
  &\ \ \ \ \ W=\sum_{M_j\in\mathcal{M}}H_{M_j}\otimes\mathbb{I}_{\overline{M_j}}.
  \end{split}
\end{equation}
The only difference between this and the previous SDP in Eq.~\eqref{SDP:ent} is that it requires the witness $W$ to be decomposable with respect to all non-trivial partition $S\in[n]$, rather than only one partition. In this scenario, any mixture of biseparable states with respect to various bipartitions, as well as mixtures of PPT states, cannot pass the SDP test; therefore, the operator $W$ obtained from the SDP constitutes a GME witness. When the minimum of the SDP with marginal set $\mathcal{M}$ is negative for state $\rho$, it provides an upper bound for the GME detection length $l_{GME}(\rho) \leq l(\mathcal{M})$.

Similarly, when we focus on a given partition $G$, the SDP can be adapted to evaluate the bipartitioned entanglement detection length: \begin{equation}
  \begin{split}
    &\min \tr(W\rho)\\
    &\exists \ P\geq0,Q\geq0\\
    &\text{s.t. }W=P+Q^{T_G},\tr(W)=d,\\
  &\ \ \ \ \ W=\sum_{M_j\in\mathcal{M}}H_{M_j}\otimes\mathbb{I}_{\overline{M_j}},
  \end{split}
\end{equation}
where $T_G$ is the partial transpose of the given partition $G$.

\section{Noise analyses for local depolarizing noise}\label{appdix:local}
In this section, we present numerical results pertaining to the noise tolerance parameter of the witnesses for both entanglement and GME across various detection lengths and marginal settings, as illustrated in Table.~\ref{table:sdp_result_ent}.

\begin{table}[!htp]
\begin{tabular}{|c|c|c|c|}
\hline
Detected state $\rho$& Marginals $\cM$& $\Tilde{p}_{Ent}(W(\cM),\rho)$&$\Tilde{p}_{GME}(W(\cM),\rho)$\\
\hline
\multirow{4}{*}{$\ket{W_3}$}&\{\{12\}\}&\multirow{2}{*}{0.4518}&\#\\
\cline{2-2}\cline{4-4}
&\{\{12\},\{23\}\}& &0.1859\\
\cline{2-4}
&\{\{12\},\{23\},\{13\}\}&0.5101&0.3039\\
\cline{2-4}
&\{\{$123$\}\}&0.7904&0.5210~\cite{jungnitsch2011taming}\\
\hline
\multirow{2}{*}{$\ket{GHZ_3^*}$}&collection of all $2$-length marginals&\#&\#\\
\cline{2-4}
&\{\{$123$\}\}&0.8&0.5714~\cite{jungnitsch2011taming}\\
\hline
\multirow{5}{*}{$\ket{W_4}(\ket{Dicke_4^2})$}&any one $2$-length marginal&\multirow{2}{*}{0.2929 (2/5)}&\multirow{2}{*}{\#}\\
\cline{2-2}
&any two distinct $2$-length marginals& &\\
\cline{2-4}
&\{\{12\},\{23\},\{34\}\}&0.2929(2/5)&0.0696 (0.0946)\\
\cline{2-4}
&\{\{23\},\{34\},\{24\}\}&0.3139 (0.4939)&0.0920 (0.1293)\\
\cline{2-4}
&collection of all $2$-length marginals&0.3508 (0.5232)&0.1541 (0.3131)\\
\hline
\multirow{8}{*}{$\ket{Cluster_4}(\ket{Ring_4})$}&collection of all $2$-length marginals&\#(\#)&\#(\#)\\
\cline{2-4}
&any one distinct $3$-length marginals& \multirow{4}{*}{2/3 (2/3)}&\#(\#)\\
\cline{2-2}\cline{4-4}
&\{\{123\},\{124\}\}& &\# (1/3)\\
\cline{2-2}\cline{4-4}
&\{\{123\},\{134\}\}& &1/3 (\#)\\
\cline{2-2}\cline{4-4}
&\{\{124\},\{134\}\} and \{\{123\},\{234\}\}& & 1/3 (1/3)\\
\cline{2-4}
&any three distinct $3$-length marginals&\multirow{2}{*}{4/5 (4/5)}&2/5(2/5)\\
\cline{2-2}\cline{4-4}
&collection of all $3$-length marginals& &1/2(1/2)\\
\cline{2-4}
&\{\{1234\}\}&8/9 (8/9)& 0.6154~\cite{jungnitsch2011taming} (0.6154)\\
\hline
\end{tabular}
\captionsetup{justification=raggedright}
\caption{This table presents noise tolerance parameters $\Tilde{p}_{Ent}$ and $\Tilde{p}_{GME}$ for entanglement and GME respectively, for some typical states under measurement involving different sets of marginals. The notation \# means that the SDP cannot construct a witness associated with given marginals to detect the entanglement or GME.}
\label{table:sdp_result_ent}
\end{table}

As the definition indicates, the EW $W(\cM)$, corresponding to the marginal set $\cM$ with a length $l(\cM)<l_{Ent}(\rho)$, is unable to detect the entanglement in state $\rho$. Analyzing the numerical results reveals a trade-off between the witness's detection length, marginal settings, and noise tolerance parameters.

In practical experiments, the measurement of every single qubit may incur an error with the same probability, says $p$. The probability of measuring $k$ qubits without error is $(1-p)^k$. From this viewpoint, local measurements are generally less susceptible to noise compared to global ones. Therefore, rather than global depolarizing noise, the local depolarizing channel describes the experimental noise more precisely \begin{equation}
    \Lambda_{ld}(p)=\Lambda_{dep}(p)^{\otimes n},
\end{equation}
which can be viewed as the depolarizing noise applied to each local qubit. Formally, the local depolarizing noise for $n$-qubit state $\rho$ can be expressed using Kraus' operators: \begin{equation}
    \Lambda_{ld}(p)(\rho)=\sum_{i,m}M_i(m,p)\rho M_i(m,p)^\dagger,
\end{equation}
where \begin{equation}
    M_i(m,p)=\sqrt{\left(1-\frac{3p}{4}\right)^m\left(\frac{p}{4}\right)^{n-m}}\bigotimes_{q=1}^n\sigma_{i_q},
\end{equation}
$\sigma_{i_q}$ is Pauli matrix $\{\sigma_0,\sigma_1,\sigma_2,\sigma_3\}$ and $m$ is the amount of $\sigma_0$ operators present in a particular $M_c$. In experiments, the implementation of measurements may incur errors. In this sense, the local depolarizing noise can reflect the degree of measurement accuracy for individual qubits. For example, to emulate the scenario that measurement on individual qubits exists error with probability $p$, one can equivalently apply local depolarizing noise with parameter $p$.

For this form of local depolarizing noise, we calculate the noise tolerance parameter related to both entanglement and GME detection. 
We compare the noise tolerance parameter related to both entanglement and GME detection across witnesses with different lengths, where each instance uses the same witness derived by SDP, and compare the tolerance for these two different kinds of noise, as presented in Table \ref{table:local_dep} and Figure.~\ref{fig:GDN+LDN}. Unsurprisingly, a witness with a longer length is more tolerant of both types of noise due to its better capability to detect entanglement. However, compared to the global depolarizing noise tolerance parameter which is greatly improved for the longer witness, the local depolarizing noise tolerance parameter grows less. Consider the GME noise tolerance parameter for $\ket{W_3}$ state as an example, as shown in the last two columns of Table.~\ref{table:local_dep}, the GME noise tolerance parameter corresponding to the global depolarizing noise increases by $71.7\%$ for the witness of $3$-length compared to $2$-length. In contrast, for local depolarizing noise, it increases by only $24.2\%$. Similar results are obtained for other states.

\begin{table}[!htp]
\begin{tabular}{|c|c|c|c|c|c|}
\hline
Detected state&Marginals of witness&$p^{Ent}_{glo}$&$p^{Ent}_{loc}$&$p^{GME}_{glo}$&$p^{GME}_{loc}$\\
\hline
\multirow{2}{*}{$\ket{W_3}$}&$\{\{12\},\{23\},\{13\}\}$&0.5101&0.3173&0.3034&0.1812\\
\cline{2-6}
&$\{\{123\}\}$&0.7904&0.4244&0.521&0.225\\
\hline
\multirow{3}{*}{$\ket{W_4}$}&collection of all $2$-length marginals&0.3508&0.2614&0.1541&0.1057\\
\cline{2-6}
&collection of all $3$-length marginals&0.776&0.3642&0.4261&0.1425\\
\cline{2-6}
&$\{\{1234\}\}$&0.8889&0.4433&0.5265&0.155\\
\hline
\multirow{3}{*}{$\ket{Dicke_4^2}$}&collection of all $2$-length marginals&0.5232&0.3095&0.3131&0.1712\\
\cline{2-6}
&collection of all $3$-length marginals&0.5232&0.3095&0.3131&0.1712\\
\cline{2-6}
&$\{\{1234\}\}$&0.8889&0.4226&0.5391&0.22\\
\hline
\multirow{2}{*}{$\ket{Cluster_4}$}&collection of all $3$-length marginals&0.8&0.3881&0.5&0.2\\
\cline{2-6}
&$\{\{1234\}\}$&0.8889&0.4684&0.6154&0.2269\\
\hline
\end{tabular}
\captionsetup{justification=raggedright}
\caption{This table presents noise tolerance parameters for entanglement and GME under both global and local depolarizing noise scenarios, for several canonical states within the context of measurements involving distinct marginal sets.}
\label{table:local_dep}
\end{table}

\begin{figure}[!h]
    \centering
    \subfigure[]{ 
            \label{} 
            \includegraphics[width=5.8cm]{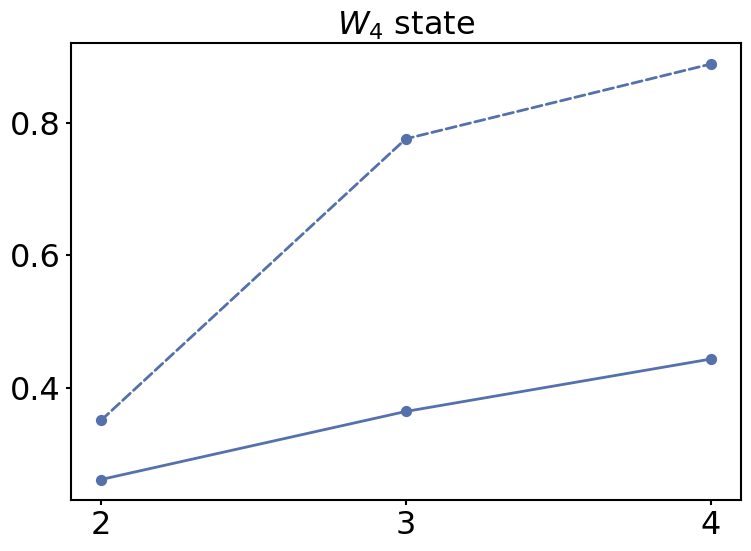}}
    \subfigure[]{ 
            \label{} 
            \includegraphics[width=5.8cm]{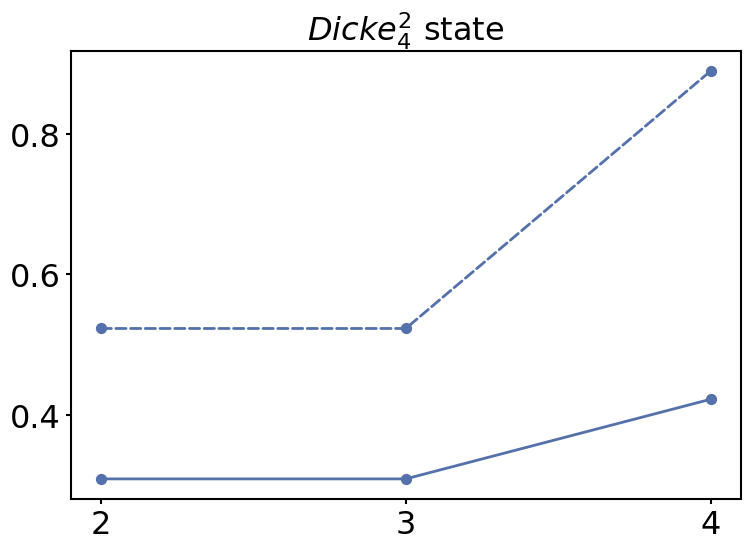}}
    \subfigure[]{ 
            \label{} 
            \includegraphics[width=5.8cm]{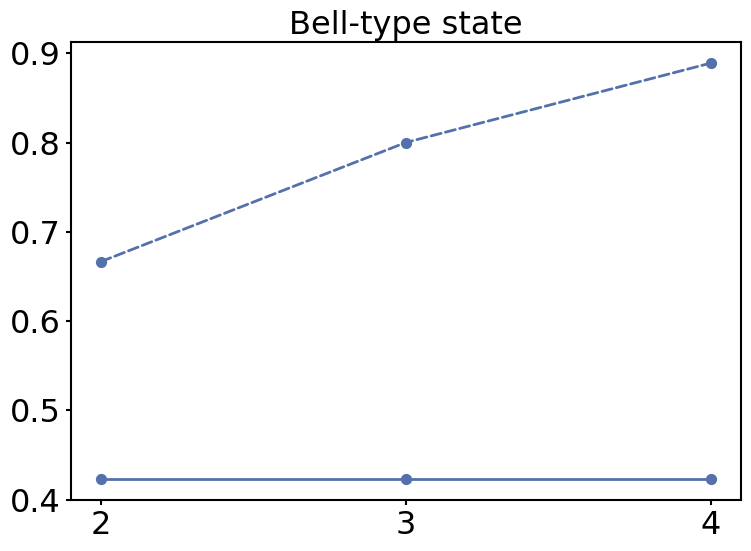}} 

\captionsetup{justification=raggedright}
    \caption{Comparison of noise tolerance among different witness lengths. We illustrate that the noise tolerance for both global depolarizing noise (GDN) and local depolarizing noise (LDN) yields similar results.}
    \label{fig:GDN+LDN}
\end{figure}

\section{Practical application of SDP}\label{appdix:prac}
In this part, we present other practical applications of the SDP method.

For $2$-dimension multipartite ($3$-quibt) Werner state, since it is entangled if and only if it violates the PPT separability criterion \cite{Eggeling2000Werner}, we could determine its entanglement detection length through SDP.
Specifically, a multipartite Werner state $\rho \in \cH_{d}^{\otimes n}$, which satisfy \begin{equation}
  \rho=(U\otimes U\ldots\otimes U)\rho(U^\dagger\otimes U^\dagger\ldots\otimes U^\dagger)
\end{equation}
for all unitary operators $U$, can be determined by $(n!-1)$ parameters $\{\alpha_1,\alpha_2,\ldots\}$, expressed as \begin{equation}
  \ketbra{Werner_n}{Werner_n}= \frac{\mathbb{I}-\sum_{i=1}^{n!-1}\alpha_i*P_{i+1}}{\tr\left(\mathbb{I}-\sum_{i=1}^{n!-1}\alpha_i*P_{i+1}\right)},
\end{equation}
where $P_i$ is the unitary operator that permutes the order of subsystems according to the $i$-th permutation when written in lexicographical order. For example, the lexicographical order of a tripartite system is $\{(1 2 3), (1 3 2), (2 1 3), (2 3 1), (3 1 2),$ $(3 2 1)\}$, so $P_5$ in this case is the operator that permutes these $3$ qubits according to $(312)$. Through numerical simulation, we find the $3$-qubit Werner state with parameter $\{0.7,0.7,0.7,0.7,0.7\}$ has entanglement detection length $l_{Ent}=2$; while the $4$-qubit Werner state with parameters $\{0.7,\ldots,0.7\}$ has entanglement length $l_{Ent}=4$. 

Besides, for the state 
 $\rho=p\ketbra{\psi_3}{\psi_3}+(1-p)\ketbra{{\ghz}_3}{{\ghz}_3}$ with $\ket{\psi_3}=(\ket{001}+\ket{010})/\sqrt{2}$ and $0<p<1$ described in proposition \ref{prop:ent_GME_gap}, a $2$-length witness to detect its bipartitioned entanglement over $G=\{1,2\}$ is described as: \begin{equation}\label{eq:bip_witness}
    W=\mathbb{I}+\frac{p}{2}\left(Z^{(2)}+Z^{(3)}-2Z^{(1)}\right)-\left(1-\frac{p}{2}\right)\left(Z^{(1)}Z^{(2)}+Z^{(1)}Z^{(3)}\right)+(1-p)Z^{(2)}Z^{(3)}-cX^{(2)}X^{(3)}
  \end{equation}
  where $c$ is a positive constant to make $W$ a decomposable PPT entanglement witness, which make $W$ expressible as \begin{equation}
    W=P+Q^{T_G}
  \end{equation}
  with $P,Q\geq0$. For example, the bipartitioned entanglement over $G=\{1,2\}$ from state $\rho=p\ketbra{\psi_3}{\psi_3}+(1-p)\ketbra{{\ghz}_3}{{\ghz}_3}$ with parameter $p=0.1$, can be detected from the observable $W$ defined in Eq.~\eqref{eq:bip_witness} with parameter $c=4.23\times 10^{-4}$.

Moreover, the symmetric state, which is either fully separable or genuine multipartite entangled \cite{Ichikawa2008sym}, can be represented by treating the Dicke state as a basis: \begin{equation}
  \rho=\sum_{i,j=0}^na_{ij}\ket{Dicke_n^i}\bra{Dicke_n^j}.
\end{equation}
Let $A$ denote the matrix formed by the coefficient $A_{ij}=a_{ij}$, which uniquely corresponds to a state. Utilizing the SDP, we provide a $2$-length entanglement witness for the $3$-qubit symmetric state $\rho(A^{(1)})$ with coefficients $A^{(1)}=\begin{pmatrix}
  1&1&1\\
  1&1&1\\
  1&1&1
\end{pmatrix}/3$, and a $3$-length entanglement witness for $\rho(A^{(2)})$ with $A^{(2)}=\begin{pmatrix}
  1&0&0\\
  0&1&0\\
  0&0&1
\end{pmatrix}/3$. Note that for a symmetric state with coefficient $a_{ij}=0$ for $i\neq j$, i.e. in the form of $\rho=\sum_ia_{ii}\ketbra{Dicke_n^i}{Dicke_n^i}$, also referred to as a diagonal symmetric state, it is separable if and only if it is PPT under the partial transpose of $\fc{n}{2}$ system \cite{shi2024entanglement}. Therefore our SDP which utilizes the PPT criterion can exactly determine the entanglement length for an arbitrary diagonal symmetric state, just like for Werner states.
Besides, since all symmetric entangled states are also genuinely entangled states \cite{Ichikawa2008sym}, we determine that $l_{Ent}(\rho(A^{(1)}))=l_{GME}(\rho(A^{(1)}))=2$ and $l_{Ent}(\rho(A^{(2)}))=l_{GME}(\rho(A^{(2)}))=3$ according to our numerical results.

\section{Concrete construction for entanglement witness and Bell's inequality}\label{appdix:construct}
In this section, we present several entanglement witnesses for diverse quantum states, serving as exemplars for the experimental deployment of observables with constrained detection lengths.

\subsection{Entanglement witness}
Concrete construction for entanglement witness $W^{Ent}(\rho)$:
\begin{itemize}
    \item $\ket{W_3}$: \begin{itemize}
        \item Marginal=\{\{12\}\}: \\
$(0.125(II+ZZ)-0.0559(IZ+ZI)-0.1118(XX+YY))\otimes I$
        \item Marginal=\{\{12\},\{23\},\{13\}\}: \\
$0.125I_8-0.0236(IZI+IIZ)-0.037ZII-0.0658(XIX+YIY+XXI+YYI)-0.017I(XX+YY)+0.0843(ZZI+ZIZ)+0.054IZZ$
        \item Marginal=\{\{123\}\}: \\
        $(2I_8+IIZ+IZZ+XIX-XZX+YIY-YZY+ZII+ZZI+2ZZZ)/16-0.08839(IXX+IYY+XXI+XXZ+YYI+YYZ+ZXX+ZYY)$
    \end{itemize}
    \item $\ket{W_4}$: \begin{itemize}
        \item Marginal=\{\{12\}\}: \\
$(I_4+ZZ)\otimes II/16-\sqrt{2}(IZ+XX+YY+ZI)\otimes II/32$
        \item Marginal=\{\{23\},\{34\},\{24\}\}: \\
        $(I_{16}-0.3913(Z^{(2)}+Z^{(3)}+X^{(2)}X^{(4)}+Y^{(2)}Y^{(4)}+X^{(3)}X^{(4)}+Y^{(3)}Y^{(4)})-0.6964Z^{(4)}+0.0648(X^{(2)}X^{(3)}+Y^{(2)}Y^{(3)})+0.1741Z^{(2)}Z^{(3)}+0.54389(Z^{(2)}Z^{(4)}+Z^{(3)}Z^{(4)}))/16$
        \item Marginal=\{\{12\},\{23\},\{34\},\{24\},\{14\}\}: \\
        $(I_{16}-0.5727(Z^{(2)}+Z^{(4)})-0.1597(X^{(1)}X^{(2)}+Y^{(1)}Y^{(2)}+X^{(3)}X^{(4)}+Y^{(3)}Y^{(4)})-0.1938(Z^{(1)}+Z^{(3)})+0.2853(Z^{(1)}Z^{(2)}+Z^{(3)}Z^{(4)})+0.0349(X^{(2)}X^{(3)}+Y^{(2)}Y^{(3)}+X^{(1)}X^{(4)}+Y^{(1)}Y^{(4)})+0.1174(Z^{(2)}Z^{(3)}+Z^{(1)}Z^{(4)})-0.4656(X^{(2)}X^{(4)}+Y^{(2)}Y^{(4)})+0.6007Z^{(2)}Z^{(4)})/16$
        \item Marginal=\{\{12\},\{23\},\{34\},\{24\},\{14\},\{13\}\}: \\
        $(I_{16}-0.4134(Z^{(1)}+Z^{(2)}+Z^{(3)}+Z^{(4)})-0.2217(X^{(1)}X^{(2)}+Y^{(1)}Y^{(2)}+X^{(2)}X^{(3)}+Y^{(2)}Y^{(3)}+X^{(3)}X^{(4)}+Y^{(3)}Y^{(4)}+X^{(4)}X^{(1)}+Y^{(4)}Y^{(1)})+0.0866(X^{(1)}X^{(3)}+Y^{(1)}Y^{(3)}+X^{(2)}X^{(4)}+Y^{(2)}Y^{(4)})+0.3614(Z^{(1)}Z^{(2)}+Z^{(2)}Z^{(3)}+Z^{(3)}Z^{(4)}+Z^{(4)}Z^{(1)}))/16$
    \end{itemize}
    \item $\ket{Dicke_4^2}$: \begin{itemize}
        \item Marginal=\{\{12\}\}: \\
        $(I_{8}-XX-YY+ZZ)\otimes II/16$
        \item Marginal=\{\{23\},\{34\},\{24\}\}: \\
        $(I+0.5774(Z^{(3)}Z^{(4)}-X^{(2)}X^{(3)}-Y^{(2)}Y^{(3)}-X^{(2)}X^{(4)}-Y^{(2)}Y^{(4)})+0.2113(X^{(3)}X^{(4)}+Y^{(3)}Y^{(4)})+0.7887(Z^{(2)}Z^{(3)}+Z^{(2)}Z^{(4)}))/16$
        \item Marginal=\{\{12\},\{23\},\{34\},\{24\},\{14\},\{13\}\}: \\
    $(I+0.0813(X^{(1)}X^{(2)}+Y^{(1)}Y^{(2)}+X^{(2)}X^{(4)}+Y^{(2)}Y^{(4)}+X^{(1)}X^{(4)}+Y^{(1)}Y^{(4)})+0.2704(Z^{(1)}Z^{(2)}+Z^{(2)}Z^{(4)}+Z^{(1)}Z^{(4)})-0.3963(X^{(2)}X^{(3)}+Y^{(2)}Y^{(3)}+X^{(3)}X^{(4)}+Y^{(3)}Y^{(4)}+X^{(1)}X^{(3)}+Y^{(1)}Y^{(3)})+0.567(Z^{(2)}Z^{(3)}+Z^{(3)}Z^{(4)}+Z^{(1)}Z^{(3)})/16$
    \end{itemize}
    \item $\ket{Cluster_4}$: \begin{itemize}
        \item Marginal=\{\{123\}\}: \\
    $(I_{8}-XZI-YYZ-ZXZ)\otimes I/16$
        \item Marginal=\{\{123\},\{234\},\{134\}\}: \\
    $(I_{16}+XZII-YYZI-ZXZI-IZXZ-IZYY-XIXZ-XIYY)/16$
        \item Marginal=\{\{1234\}\}: \\
    $(I_{16}-IIZX+XZII-IZXZ-IZYY-XIXZ-XIYY-YYIX-ZXIX+YYZI+ZXZI-XZZX-YXXY+YXYZ+ZYXY-ZYYZ)/16$
    \end{itemize}
\end{itemize}

\subsection{GME witness}
Concrete construction for GME witness $W^{GME}(\rho)$: \begin{itemize}
    \item $\ket{W_3}$: \begin{itemize}
        \item Marginal=\{\{12\},\{23\}\}: \\
$0.125I_8-0.0218(ZII+IIZ)-0.0518IZI-0.0428(XXI+YYI+IXX+IYY)+0.0112(IZZ+ZZI)$
        \item Marginal=\{\{12\},\{23\},\{13\}\}: \\
$0.125I_8-0.0314(ZII+IZI+IIZ)-0.0297(IXX+IYY+XIX+YIY+XXI+YYI)+0.0293(IZZ+ZIZ+ZZI)$
    \end{itemize}
    \item $\ket{W_4}$: \begin{itemize}
        \item Marginal=\{\{12\},\{23\},\{34\}\}: \\
        $(I-0.2514(Z^{(2)}+Z^{(3)})-0.1341(Z^{(1)}+Z^{(4)})-0.2176(X^{(1)}X^{(2)}+Y^{(1)}Y^{(2)}+X^{(3)}X^{(4)}+Y^{(3)}Y^{(4)})-0.2542(X^{(2)}X^{(3)}+Y^{(2)}Y^{(3)})-0.0052(Z^{(1)}Z^{(2)}+Z^{(3)}Z^{(4)})-0.885Z^{(2)}Z^{(3)})/16$
        \item Marginal=\{\{12\},\{13\},\{14\}\}: \\
        $(I-0.4974Z^{(1)}-0.133(Z^{(2)}+Z^{(3)}+Z^{(4)})-0.2177(X^{(1)}X^{(2)}+Y^{(1)}Y^{(2)}+X^{(1)}X^{(3)}+Y^{(1)}Y^{(3)}+X^{(1)}X^{(4)}+Y^{(1)}Y^{(4)})+0.0318(Z^{(1)}Z^{(2)}+Z^{(1)}Z^{(3)}+Z^{(1)}Z^{(4)}))/16$
        \item Marginal=\{\{12\},\{23\},\{34\},\{14\}\}: \\
        $(I-0.22626(Z^{(1)}+Z^{(2)}+Z^{(3)}+Z^{(4)})-0.1548(X^{(1)}X^{(2)}+Y^{(1)}Y^{(2)}+X^{(2)}X^{(3)}+Y^{(2)}Y^{(3)}+X^{(3)}X^{(4)}+Y^{(3)}Y^{(4)}+X^{(1)}X^{(4)}+Y^{(1)}Y^{(4)})+0.078(Z^{(1)}Z^{(2)}+Z^{(2)}Z^{(3)}+Z^{(3)}Z^{(4)}+Z^{(1)}Z^{(4)}))/16$
        \item Marginal=\{\{12\},\{23\},\{34\},\{24\}\}: \\
        $(I-0.4167Z^{(2)}-0.1428Z^{(1)}-0.1782(Z^{(3)}+Z^{(4)})-0.2788(X^{(1)}X^{(2)}+Y^{(1)}Y^{(2)})-0.0104Z^{(1)}Z^{(2)}-0.1561(X^{(2)}X^{(3)}+Y^{(2)}Y^{(3)}+X^{(2)}X^{(4)}+Y^{(2)}Y^{(4)})+0.0429(Z^{(2)}Z^{(3)}+Z^{(2)}Z^{(4)})-0.0623(X^{(3)}X^{(4)}+Y^{(3)}Y^{(4)})+0.061Z^{(3)}Z^{(4)})/16$
        \item Marginal=\{\{12\},\{23\},\{34\},\{24\},\{14\}\}: \\
        $(I-0.303(Z^{(2)}+Z^{(4)})-0.219(Z^{(2)}+Z^{(4)})-0.1634(X^{(1)}X^{(2)}+Y^{(1)}Y^{(2)}+X^{(2)}X^{(3)}+Y^{(2)}Y^{(3)}+X^{(3)}X^{(4)}+Y^{(3)}Y^{(4)}+X^{(1)}X^{(4)}+Y^{(1)}Y^{(4)})+0.0503(Z^{(1)}Z^{(2)}+Z^{(2)}Z^{(3)}+Z^{(3)}Z^{(4)}+Z^{(1)}Z^{(4)})+0.0238(X^{(2)}X^{(4)}+Y^{(2)}Y^{(4)})+0.1544Z^{(1)}Z^{(4)})/16$
        \item Marginal=\{\{12\},\{23\},\{34\},\{14\},\{13\},\{24\}\}: \\
        $(I-0.2802(Z^{(1)}+Z^{(2)}+Z^{(3)}+Z^{(4)})-0.1037(X^{(1)}X^{(2)}+Y^{(1)}Y^{(2)}+X^{(2)}X^{(3)}+Y^{(2)}Y^{(3)}+X^{(3)}X^{(4)}+Y^{(3)}Y^{(4)}+X^{(1)}X^{(4)}+Y^{(1)}Y^{(4)}+X^{(1)}X^{(3)}+Y^{(1)}Y^{(3)}+X^{(2)}X^{(4)}+Y^{(2)}Y^{(4)})+0.0919(Z^{(1)}Z^{(2)}+Z^{(2)}Z^{(3)}+Z^{(3)}Z^{(4)}+Z^{(1)}Z^{(4)}+Z^{(2)}Z^{(4)}+Z^{(1)}Z^{(3)}))/16$
    \end{itemize}
    \item $\ket{Dicke_4^2}$: \begin{itemize}
        \item Marginal=\{\{12\},\{23\},\{34\}\}: \\
        $(I-0.13(X^{(1)}X^{(2)}+Y^{(1)}Y^{(2)}+X^{(3)}X^{(4)}+Y^{(3)}Y^{(4)})-0.5659X^{(2)}X^{(3)}+Y^{(2)}Y^{(3)}+0.1838(Z^{(1)}Z^{(2)}+Z^{(3)}Z^{(4)})-0.3579Z^{(2)}Z^{(3)})/16$
        \item Marginal=\{\{12\},\{13\},\{14\}\}: \\
        $(I-0.2711(X^{(1)}X^{(2)}+Y^{(1)}Y^{(2)}+X^{(1)}X^{(3)}+Y^{(1)}Y^{(3)}+X^{(1)}X^{(4)}+Y^{(1)}Y^{(4)})+0.0641(Z^{(1)}Z^{(2)}+Z^{(1)}Z^{(3)}+Z^{(1)}Z^{(4)}))/16$
        \item Marginal=\{\{12\},\{23\},\{34\}\}: \\
        $(I-0.216(X^{(1)}X^{(2)}+Y^{(1)}Y^{(2)}+X^{(2)}X^{(3)}+Y^{(2)}Y^{(3)}+X^{(3)}X^{(4)}+Y^{(3)}Y^{(4)}+X^{(1)}X^{(4)}+Y^{(1)}Y^{(4)})+0.0266(Z^{(1)}Z^{(2)}+Z^{(2)}Z^{(3)}+Z^{(3)}Z^{(4)}+Z^{(1)}Z^{(4)})$
        \item Marginal=\{\{12\},\{23\},\{34\},\{24\}\}: \\
        $(I-0.1863(X^{(2)}X^{(3)}+Y^{(2)}Y^{(3)}+X^{(2)}X^{(4)}+Y^{(2)}Y^{(4)})-0.0416(X^{(3)}X^{(4)}+Y^{(3)}Y^{(4)})-0.3057(X^{(1)}X^{(2)}+Y^{(1)}Y^{(2)})+0.2305(Z^{(2)}Z^{(3)}+Z^{(2)}Z^{(4)})+0.2116Z^{(3)}Z^{(4)}+0.0851Z^{(1)}Z^{(2)}$
        \item Marginal=\{\{12\},\{23\},\{34\},\{24\},\{14\}\}: \\
        $(I-0.1287(X^{(1)}X^{(2)}+Y^{(1)}Y^{(2)}+X^{(2)}X^{(3)}+Y^{(2)}Y^{(3)}+X^{(3)}X^{(4)}+Y^{(3)}Y^{(4)}+X^{(1)}X^{(4)}+Y^{(1)}Y^{(4)})+0.2403(Z^{(1)}Z^{(2)}+Z^{(2)}Z^{(3)}+Z^{(3)}Z^{(4)}+Z^{(1)}Z^{(4)})-0.1551(X^{(2)}X^{(4)}+Y^{(2)}Y^{(4)})+0.3132Z^{(2)}Z^{(4)})/16$
        \item Marginal=\{\{12\},\{23\},\{34\},\{24\},\{14\},\{13\}: \\
        $(I-0.1228(X^{(1)}X^{(2)}+Y^{(1)}Y^{(2)}+X^{(2)}X^{(3)}+Y^{(2)}Y^{(3)}+X^{(3)}X^{(4)}+Y^{(3)}Y^{(4)}+X^{(1)}X^{(4)}+Y^{(1)}Y^{(4)}+X^{(2)}X^{(4)}+Y^{(2)}Y^{(4)}+X^{(1)}X^{(3)}+Y^{(1)}Y^{(3)})+0.2368(Z^{(1)}Z^{(2)}+Z^{(2)}Z^{(3)}+Z^{(3)}Z^{(4)}+Z^{(1)}Z^{(4)}+Z^{(2)}Z^{(4)}+Z^{(1)}Z^{(3)}))/16$
    \end{itemize}
    \item $\ket{Cluster_4}:$ \begin{itemize}
        \item Marginal=\{\{124\},\{134\}\}: \\
        $(4I_{16}-XZII+YYIX+ZXIX+IIZX+XIXZ+XIYY)/64$
        \item Marginal=\{\{123\},\{234\},\{124\}\}: \\
        $(I-0.2092XZII-0.2287(YYZI+ZXZI+YYIX+ZXIX)+0.1241IIZX-0.3333(IZXZ+IZYY))/16$
        \item Marginal=\{\{123\},\{234\},\{124\},\{134\}\}: \\
        $(I+0.934(XZII+IIZX)-0.2734(YYZI+ZXZI+YYIX+ZXIX+XIXZ+XIYY+IZXZ+IZYY))/16$
        \item Marginal=\{\{1234\}\}: \\
        $(5I_{16}+IIZX-IZXI-IZYY-XIXZ-XIYY+XZII-3XZZX-YXXY+YXYZ-YYIX-YYZI-ZXIX-ZXZI+ZYXY-ZYYZ)/80$
    \end{itemize}
    \item $\ket{Ring_4}:$ \begin{itemize}
        \item Marginal=\{\{123\},\{234\}\}: \\
        $(4I_{16}+YXYI-XIXI-ZXZI-IXIX=IYXY-IZXZ)/64$
        \item Marginal=\{\{123\},\{234\},\{124\}\}: \\
        $(I+0.1241XIXI+0.333(YXYI-ZXZI)-0.2092IXIX+0.2287(IYXY-IZXZ+XYIY-XZIZ))/16$
        \item Marginal=\{\{1234\}\}: \\
        $(5I_{16}+IXIX+IYXY-IZXZ+XIXI-3XXXX+XYIY-XZIZ+YIYX+YXYI-YYZZ-YZZY-ZIZX-ZXZI-ZYYZ-ZZYY)/80$
    \end{itemize}
\end{itemize}

\bibliographystyle{apsrev4-1}
\bibliography{entangle_length.bib}
\end{document}